\documentclass[reqno,a4paper,12pt]{amsart}

\usepackage[hmargin=1in,vmargin=1in]{geometry}
\usepackage{amssymb,amsfonts,amsmath,amsthm}
\usepackage{amstext}
\usepackage{graphicx}

\usepackage{float}

\numberwithin{equation}{section}

\newtheorem{theorem}{Theorem}[section]

\newtheorem{lemma}[theorem]{Lemma}

\theoremstyle{definition}

\newtheorem{remark}[theorem]{Remark}

\newcommand{\rset}{\mathbb{R}}
\newcommand{\Rd}{{\rm d}}
\newcommand{\bB}{\bar{B}} %
\newcommand{\bA}{\bar{A}} %
\newcommand{\dD}{{\rm D}}

\newcommand{\Si}{\mathrm{Si}}
\newcommand{\prt}[1]{\left( #1\right)}

\title[Langevin dynamics from quantum mechanics]{Classical Langevin dynamics derived from quantum mechanics} 

\author{H{\aa}kon Hoel}
\address{
Computer, Electrical and Mathematical Science and Engineering Division, 
King Abdullah University of Science and Technology, Thuwal}
\email{haakonah1@gmail.com, hakon.hoel@kaust.edu.sa}
\author{Anders Szepessy}
\address{Institutionen f\"or Matematik, Kungl. Tekniska H\"ogskolan, 100 44 Stockholm, Sweden}
\email{szepessy@kth.se}
\thanks{The research was supported by
Swedish Research Council 621-2014-4776} 

\date{} 
\begin{document}

\begin{abstract}
The classical work by Zwanzig [J. Stat. Phys. 9 (1973) 215-220] derived Langevin dynamics from a Hamiltonian system of a heavy particle coupled to a heat bath. This work extends Zwanzig's model to a quantum system and formulates a more general coupling between a particle system and a heat bath. The main result proves that {\it ab initio} Langevin molecular dynamics, with a certain rank one friction matrix determined by the coupling, approximates for any temperature canonical quantum observables, based on the system coordinates,  more accurately than any Hamiltonian system in these coordinates, for large mass ratio between the system and the heat bath nuclei. 
\end{abstract}

\maketitle
\tableofcontents

\section{Langevin molecular dynamics}

Langevin dynamics for (unit mass) particle systems with position coordinates 
$X_L:[0,\infty)\to\rset^N$ and momentum coordinates $P_L:[0,\infty)\to \rset^N$ , 
defined by 
\begin{equation}\label{Ito_langevin2}
\begin{split}
{\rm d}X_L(t) &= P_L(t){\rm d}t\\
{\rm d}P_L(t) &= -\nabla \lambda\big(X_L(t)\big){\rm d}t - \bar\kappa P_L(t){\rm d}t  + (2 \bar\kappa T)^{1/2}\ {\rm d} W(t)\, ,
\end{split}
\end{equation}
is used for instance to simulate molecular dynamics in the canonical ensemble
of constant temperature $T$, volume and number of particles, where $W$ denotes the standard  
Wiener process with $N$ independent components.
The purpose of this work is 
to precisely determine both the potential $\lambda:\mathbb R^N\to \mathbb R$ and the friction matrix  
$\bar\kappa\in \mathbb R^{N\times N}$ in this equation, from a quantum mechanical
model of a molecular system including weak coupling to a heat bath.

Molecular systems are described by the Schr\"odinger equation
with a potential based on Coulomb interaction of all nuclei and electrons in the system. 
This quantum mechanical model
is complete in the sense that no unknown parameters enter - the observables in the canonical ensemble
are determined from the Hamiltonian and the temperature. The classical limit of the quantum formulation,
yields an accurate approximation of the observables based on the nuclei only,
for large nuclei--electron mass ratio $M$. Ab initio molecular dynamics based on the electron ground state eigenvalue can be used when the temperature is low compared to the first electron eigenvalue gap. 
A certain weighted average of different ab initio dynamics, corresponding to each electron eigenvalue, approximates quantum observables for any temperature, see \cite{KSS}, also in the case of  
observables including time correlation and many particles.
The elimination of the electrons provides a substantial
computational reduction, making it possible to simulate large molecular systems, cf. \cite{marx}.

 In molecular dynamics simulations one often wants to determine properties of a large macroscopic system with many particles, say $N\sim 10^{23}$. Such large particle systems cannot yet be simulated on a computer and one may then ask for a setting where a smaller system has similar properties as the large. Therefore, we seek an equilibrium density that has the property that the marginal distribution for a subsystem has the same density as the whole system. In \cite{KSS} it is motivated
 how this assumption leads to the Gibbs measure, i.e. the canonical ensemble; this is also the motivation to
 use the canonical ensemble for the composite system in this work, although some studies on heat bath models
 use the microcanonical ensemble for the composite whole system.

Langevin dynamics is often introduced
to sample initial configurations from the
Gibbs distribution and to avoid to simulate the dynamics of all heat bath particles.
 The friction/damping parameter in the Langevin equation is then typically set 
small enough to not perturb the dynamics too much and large enough to avoid long sampling times. 

The purpose of this work is to show that  Langevin molecular dynamics, for the non heat bath nuclei, with a certain friction/damping parameter determined from the Hamiltonian,
approximates  the quantum system, in the canonical ensemble for any temperature, 
more accurately than any Hamiltonian dynamics (for the non heat bath particles)
in the case the system is 
weakly coupled to a heat bath of many fast particles.

Our heat bath model is based on the assumption of weak coupling 
- in the sense that the  perturbation in the system from the heat bath is small and vice versa -
which we show leads to  Zwanzig's model for nonlinear generalized Langevin equations in \cite{zwanzig_1973}, with a harmonic oscillator heat bath.
Zwanzig also derives a pure Langevin equation: he assumes first a continuous Debye distribution
of the eigenvalues of the heat bath potential energy quadratic form; in the next step he lets  
 the coupling to the heat bath have a special form, so that the integral kernel for the friction term in the generalized Langevin equation becomes
a Dirac-delta measure. Our derivation uses a heat bath based on
nearest neighbour interaction on an infinite 
cubic lattice, so that the continuum distribution of eigenvalues is rigorously obtained by
considering a difference operator with an infinite number of nodes.
Our convergence towards a pure Langevin equation is not based on a Dirac-delta measure for the integral kernel
but obtained
from the time scale separation of the fast heat bath particles and the slower system particles,
which allows a general coupling and a bounded covariance matrix for the fluctuations. 
The fast heat bath dynamics is provided either from light heat bath particles or a stiff heat bath potential.
By a stiff heat bath we mean that the
smallest eigenvalue of
the Hessian of the heat bath potential energy is of the order $\chi^{-1}$, where $\chi\ll 1$. We show that
the friction/damping coefficient  in the pure Langevin equation is determined by the 
derivative of the forces on the heat bath particles with respect to the system particle positions.
We also prove that the observables of the system coordinates in the system-heat bath quantum model
can be approximated using this Langevin dynamics with accuracy $\mathcal O(m\log m^{-1} +(mM)^{-1})$,
where $M\gg 1$ is the system nuclei--electron mass ratio and $m\ll 1$ is the heat bath nuclei -- system nuclei mass ratio;
the approximation  by a Hamiltonian system
yields the corresponding larger error estimate $\mathcal O(m^{1/2} +(mM)^{-1})$, in the case of light heat bath particles. 
In this sense, our Langevin equation is a better approximation.
The case with a stiff heat bath system has the analogous error estimate 
$\mathcal O(\chi^{2\delta}\log\chi^{-1} +M^{-1})$ with the correct friction/damping parameter, while a
Hamiltonian system gives the larger error $\mathcal O(\chi^{2\delta-1/2} +M^{-1})$, where $\chi^\delta\ll \chi^{1/4}$
measures the coupling between the heat bath and the system. 
Our main assumptions are:
\begin{itemize}
\item the coupling between the system and the heat bath is weak and localized,  
\item either the heat bath particles are much lighter than the system particles or the heat bath is stiff,
\item the harmonic oscillator heat bath
is constructed from nearest neighbour interaction on an infinite cubic lattice in dimension three,  
\item the heat bath particles are initially randomly Gibbs distributed (conditioned by the system particle coordinates), and 
\item the system potential energy and the observables are sufficiently regular. 
\end{itemize}

The system and the heat bath is modelled by  a Hamiltonian where the system potential energy is perturbed by
$\bar V(x,X)$ with system particle positions $X\in\mathbb R^N$ and heat bath particle positions $x\in\mathbb R^n$.
Our assumption of weak coupling is formulated as the requirement that the potential $\bar V$ satisfies
\[
\min_{x\in \rset^n} \bar V(x,X)=\bar V\big(a(X),X\big)=0\, .
\]
Theorem \ref{thm1} proves that the obtained minimizer $a(X)\in \mathbb R^n$ determines the friction matrix $\bar\kappa$ in \eqref{Ito_langevin2} by the rank one $N\times N$ matrix
\[
\bar\kappa_{\ell\ell'}=\bar c m^{1/2}\langle \bar V''(a,X)\partial_{X_\ell} a,\bar V''(a,X)\partial_{X_{\ell'}}a\rangle\, ,
\]
where $\bar V$ and $a$ have limits as $n\to \infty$; here $\bar V''$ denotes the Hessian of $\bar V$ with respect to $x$, the brackets $\langle\cdot,\cdot\rangle$ denote  the  scalar product in $\mathbb R^n$ and
$\bar c$ is a constant (related to the density of states). 

In a setting when the temperature is small compared to the difference of the two smallest eigenvalues of the Hamiltonian symbol, with given nuclei coordinates, it is well known that ab initio molecular dynamics
is based on the ground state electron eigenvalue as the potential $\lambda$ in \eqref{Ito_langevin2}, cf. \cite{marx}.
When the temperature is larger, excited electron states influence the nuclei dynamics. The work \cite{KSS}
derives a molecular dynamics approximation of quantum observables, 
including time correlations, as a certain weighted average of
ab initio observables in the excited states, and these ideas are put into the context 
of this work in Section \ref{sec2}.

The analysis of particles in a heat bath has a long history, starting with the work by Einstein and Smoluchowski.
The Langevin equation was introduced in \cite{langevin} to study Brownian motion mathematically, 
before the Wiener process was available. 
Early results on the elimination of the heat bath degrees of freedom to obtain a Langevin equation are \cite{ford_kac1, ford_kac2} and \cite{zwanzig_1973},\cite[Section 9.3]{zwanzig_book}.
The work \cite{ford_kac1, ford_kac2} include in addition a derivation of a quantum Langevin equation, which is based on an operator version of the classical Langevin equation. 
Our study of the classical Langevin equation from quantum mechanics is not related to this quantum Langevin equation.
We start with an ab initio quantum model of the system coupled to a heat bath in the canonical ensemble and use its classical limit to eliminate the electron degrees of freedom. Then the classical system coupled to the heat bath is analyzed by separation of time scales.
The separation of time scales of light bath particles
and heavier system particles was first used in \cite{lebowitz_rubin}, see also 
\cite{Mazur-Oppenheim,Oppenheim-1996}
to determine a Fokker-Planck equation for the system particle,
from the Liouville equation of the coupled system using a formal expansion in the small mass ratio.
Section 8 in \cite{pavliotis}
presents a proof, and several references to related work, where the Langevin equation is derived from the generalized Langevin equation, using exponential decay of the kernel in the memory term of the generalized Langevin equation.

Our contribution employs the separation of time scales approach previously used in~\cite{lebowitz_rubin,Mazur-Oppenheim,Oppenheim-1996}, but a novelty
is that we here present mathematical proofs for the weak convergence rate of 
Langevin dynamics towards quantum mechanics.
Our work also differs from~\cite[Section 8]{pavliotis}, wherein the kernel of the generalized Langevin equation is assumed to be such that by adding a finite dimensional variable the system becomes Markovian and the kernel tends to a point mass. For instance, our kernel vanishes as $m\to 0^+$ and we use the density of  heat bath states to determine the kernel. Using the precise information from the density of states in the case of heat bath nearest neighbor interaction in a cubic lattice we obtain a positive definite friction matrix $\bar\kappa$, while if the nearest neighbor interaction would be related to a lattice in dimension four the friction matrix would vanish, see Remark \ref{rem_vanish}.

The main new ideas in our work are the first principles formulation from quantum mechanics,  the weak coupling condition as a minimization, 
the precise use of the density of heat bath states,
that  
the error estimate uses stability of the Kolmogorov backward equation for the Langevin equation
evaluated along the dynamics of the coupled system, and formulation of numerical schemes and numerical results related to Langevin dynamics approximation of particles systems.

Section~\ref{sec_model} formulates a classical model for the system and the heat bath, including the weak coupling, and derives the corresponding generalized Langevin equation, based on a memory term with a specific integral kernel, in a classical molecular dynamics setting. The generalized Langevin equation is analyzed in Section \ref{sec3}, with subsections related to the dissipation, fluctuations and approximation by pure Langevin dynamics in the case of light heat bath particles. The main result of Section \ref{sec3} is Theorem \ref{thm1}, where a certain Langevin dynamics is shown to approximate
a classical system weakly coupled to a heat bath.
Section \ref{chi_sec} extends the analysis to the case of stiff heat baths and derives a  corresponding approximation result in Theorem \ref{thm2}. Section \ref{sec2} relates the classical model to a quantum formulation and provides background and error estimates on quantum observables in the canonical ensemble approximated by classical molecular dynamics. 
The main result of the work is Theorem \ref{kombin}, which proves that for any temperature canonical quantum observables based on the system coordinates can be accurately approximates by the Langevin dynamics obtained in Theorems \ref{thm1} and \ref{thm2}. Section \ref{sec_num} includes numerical results of the system particle density and autocorrelation of the coupled system and heat bath approximated by the particle density and autocorrelation for the Langevin dynamics.

\section{The model of the system and the heat bath}\label{sec_model}
We consider in this section a classical model of
a molecular system, with position coordinates  $X\in\rset^N$ and momentum coordinates $P\in \rset^N$, coupled to a heat bath, with position and momentum coordinates $x\in \rset^n$ and $p\in\rset^n$, respectively, represented by the Hamiltonian
\[
\frac{|P|^2}{2}  + \frac{|p|^2}{2m} +\lambda(X)+\bar V(x,X)\, ,
\]
where $\lambda:\rset^N\to \mathbb R$ is the potential energy for the system
and $\bar V:\rset^n\times\rset^{N}\to\mathbb R$ is the potential energy for the heat bath including the coupling to the
system. The parameter $m$ is the mass ratio between  heat bath nuclei and system nuclei.
We have set the time scale so that the system nuclei mass is one.
In Section \ref{sec2} we show that this model is the classical limit of a quantum model and study the accuracy
of the classical approximation, also in the case with system nuclei  that have different masses.

We study small perturbations of
the equilibrium bath state $x=a(X)\in \rset^n$ where 
\[
\min_{x\in \rset^n} \bar V(x,X) = 
\bar V\big(a(X),X\big)\, .
\]
Taylor expansion around the equilibrium yields
\[
\bar V(x,X)=\bar V\big(a(X),X\big) +\frac{1}{2}\big\langle x-a(X),
\bar V''_{xx}\big(a(X)+\xi(x-a(X)),X\big)\big(x-a(X)\big)\big\rangle
\]
for some $\xi\in[0,1]$, where $\bar V_{xx}$ is the Hessian matrix  in $\mathbb R^{n\times n}$, with respect to the $x$ coordinate,
and the notation $\langle \cdot, \cdot\rangle$ is 
the standard scalar product in $\mathbb C^n$.
We assume that the coupling between the system and the heat bath is weak, which means that the perturbation in the system from the heat bath must be small. As the perturbation $|x-a(X)|\to 0$ we therefore require that
\begin{equation}\label{2}
\bar V(a(X),X)=0\, .
\end{equation}
Weak coupling also means that the perturbation in the heat bath from the system
is small, so that the influence on the Hessian $\bar V''_{xx}$ from $X$ is negligible,
and we assume therefore that
\begin{equation}\label{3}
\bar V_{xx}''\big(a(X)+\xi(x-a),X\big)=C
\end{equation}
where $C$ is constant, symmetric and positive definite.  Hence the vibration frequencies 
for $x$ are assumed to be constant for all $x$ and $X$. We use the notation $C=\bar V''$ below.

The assumptions 
\eqref{2} and \eqref{3} of weak coupling lead to the 
Hamiltonian $H:\rset^N\times\rset^N\times\rset^n\times\rset^n\to \rset$ 
\begin{equation}\label{star}
H(X,P,x,p)= 
\frac{|P|^2}{2}  + \frac{|p|^2}{2m} + \lambda(X) 
+ \frac{1}{2} \langle \bar V'' \big(x-a(X)\big), \, x-a(X) \rangle\, .
\end{equation}
This model \eqref{star} is of the same form as the model of interaction with a heat bath introduced and analysed by Zwanzig in the seminal work \cite{zwanzig_1973}, although here the motivation with weak coupling and several system particles is different.

The Hamiltonian \eqref{star} yields the dynamics
\[
\begin{split}
\dot X_t &= P_t\\
\dot P_t &= -\nabla\lambda(X_t) + \langle \bar V''\big(x_t-a(X_t)\big), \nabla a(X_t)\rangle \\
\dot x_t &= p_t/m\\
\dot p_t &= -\bar V''\big(x_t-a(X_t)\big)\, 
\end{split}
\]
and the change of variables $(m\bar V'')^{1/2} r_t:=p_t$ 
implies
\[
\begin{split}
\dot x_t &= m^{-1/2}\bar V''^{1/2} r_t\\
\dot r_t &= -m^{-1/2}\bar V''^{1/2} \big(x_t-a(X_t)\big)\, .
\end{split}
\]
Define $\varphi_t:=x_t-a(X_t)+ {\rm i} r_t$ to obtain
\begin{equation}\label{1.2}
\begin{split}
\dot X_t &= P_t\\
\dot P_t &= -\nabla \lambda(X_t) + {\rm Re}\langle\bar V''\varphi_t,\nabla a(X_t)\rangle\\
{\rm i}\dot\varphi_t &= m^{-1/2}\bar V''^{1/2} \varphi_t -{\rm  i} \dot X_t\cdot \nabla a(X_t)\, ,
\end{split}
\end{equation}
where 
the third equation uses the notation $v\cdot w$ 
for the standard scalar product in $\mathbb R^N$.

Assume that $x_0$ and $p_0$ are Gaussian with the distributions provided by the
marginals of the Gibbs density
\[
\frac{e^{-H(X,P,x,p)/T}}{\int_{\rset^{2N+2n}} e^{-H(X,P,x,p)/T} {\rm d}X{\rm d}P{\rm d}x{\rm d}p}\, ,
\]
that is, the momentum $p$ is multivariate normal distributed with mean zero and covariance matrix $mT\, {\rm I}$ and independent of
$x$, which is multivariate normal distributed with mean $a(X)$ and covariance matrix $T(\bar V'')^{-1}$.
Consequently the initial data can be written
\begin{equation}\label{phi0}
\varphi_0=\sum_{k=1}^n \gamma_k\nu_k'
\end{equation}
using the orthogonal eigenvectors $\nu_m'$, normalized as $\langle\nu'_k,\nu'_k\rangle=1$, and eigenvalues $\mu_m$ of $\bar V''$
\begin{equation}\label{egenvekt}
\begin{split}
\bar V''\nu_m' &= \mu_m\nu_m'\, ,\quad
\nu_m'\in \rset^n \mbox{ and } \mu_m\in\rset_+
\end{split}
\end{equation}
and the definition
\begin{equation}\label{gamma0}
\gamma_m:=\gamma_{m}^r+{\rm i}\gamma_{m}^i
\end{equation}
with $\gamma_{m}^r$ and $\gamma_m^i, \ m=1,\ldots, n$, independent and normal distributed real scalar random numbers with mean zero and variance $T/\mu_m$.

Duhamel's principle shows that
\[
\varphi_t= e^{-{\rm i} \bar V''^{1/2} m^{-1/2} t}\varphi_0 - \int_0^t e^{-{\rm i}\bar V''^{1/2}m^{-1/2}(t-s)} \dot a(X_s) {\rm d}s
\]
which implies a form of Zwanzig's generalized  Langevin equation
\begin{equation}\label{1.5}
\begin{split}
\ddot X_t&= -\nabla \lambda(X_t) - \int_0^t \langle
\bar V'' \cos\big(\frac{(t-s) \bar V''^{1/2}}{m^{1/2}}\big) \dot a(X_s), \nabla a(X_t)\rangle{\rm d} s\\
&\quad + {\rm Re}\langle \bar V'' e^{-{\rm i} \bar V''^{1/2} m^{-1/2}t}\varphi_0, \nabla a(X_t)\rangle\, ,
\end{split}
\end{equation}
with non Markovian friction term given by the integral and a noise term including the stochastic initial data $\varphi_0$.
We study two different cases:
\begin{itemize}
\item either the mass ratio $m\ll 1$ 
is small, or
\item the smallest eigenvalue of $\bar V''$ is large of the order $\chi^{-1}$ while the coupling
derivative $\|\nabla a\|$ is small of size $\chi^\delta$ with $\delta>1/4$.
\end{itemize}
In these cases,
both the friction and the noise terms are based on highly oscillatory functions,
which will make these contributions small, as explained in the next section. To simplify the analysis, we assume also
that 
\begin{equation}\label{Da_const}
\mbox{$\nabla a(\cdot)$ is constant.}
\end{equation}

\section{Analysis of the generalized Langevin equation for $m\ll 1$}\label{sec3}
In this section we first study the 
dissipation term $\int_0^t \langle\bar V'' \cos\big(\frac{(t-s)\bar V''^{1/2}}{m^{1/2}}\big) \dot a(X_s), \nabla a\rangle
{\rm d} s$ and the fluctuation term ${\rm Re}\langle \bar V'' e^{-{\rm i} \bar V''^{1/2} m^{1/2} t}\varphi_0, \nabla a\rangle$ in
\eqref{1.5},
as the number of bath particles tend to infinity and the mass ratio, $m$, between light heat bath nuclei and
heavier system nuclei is small.
Then in Section \ref{system_dyn} 
we prove an priori estimate of $\sup_{s<t}\mathbb E[|X_s|^2+|\dot X_s|^2+|\ddot X_s|^2]$.
Section \ref{du}   uses the 
limit terms  to construct a
Langevin equation, namely the It\^{o} stochastic differential equation
\begin{equation}\label{Ito_langevin}
\begin{split}
{\rm d}X_L(t) &= P_L(t){\rm d}t\\
{\rm d}P_L(t) &= -\nabla \lambda\big(X_L(t)\big){\rm d}t - m^{1/2}\kappa P_L(t){\rm d}t  + (2 m^{1/2}\kappa T)^{1/2}\ {\rm d} W(t)\, ,
\end{split}
\end{equation}
with a certain symmetric friction  matrix $\kappa\in \mathbb R^{N\times N}$ and a Wiener process $W:[0,\infty)\times\Omega\to \mathbb R^N$, with $N$ independent components; here $\Omega$ is the set of outcomes for the process $(X,P):[0,\infty)\times \Omega\to\mathbb R^{2N}$.
Finally, we use the solution of the Kolmogorov backward equation for the Langevin dynamics
along a solution path $(X_t,P_t)$ of \eqref{1.5} to derive an error estimate of the approximation, namely
 $\mathbb E[g(X_t,P_t)]-\mathbb E[g(X_L(t),P_L(t)]=\mathcal O(m\log m^{-1})$, for  any
given smooth bounded observable $g:\mathbb R^{2N}\to \mathbb R$ and equal initial data $(X(0),P(0))=
(X_L(0),P_L(0))$.

\subsection{The dissipation term and a precise heat bath}
The change of variables
\[
\frac{t-s}{m^{1/2}} =\tau
\]
yields
\[
\begin{split}
&\int_0^t \langle \bar V'' \cos\big(\frac{(t-s)V''^{1/2}}{m^{1/2}}\big) \dot a(X_s), \partial_{X^\ell} a\rangle{\rm d} s\\
&= m^{1/2}\int_0^{t/\sqrt{m}}  \langle\bar V'' \cos(\tau \bar V''^{1/2}) \dot a(X_{t-m^{1/2}\tau}),  \partial_{X^\ell}a\rangle {\rm d} \tau\\
&=\sum_{\ell'}
m^{1/2}\int_0^{t/\sqrt{m}}  \langle\bar V'' \cos(\tau\bar V''^{1/2}) \partial_{X^{\ell'}}a, \partial_{X^{\ell}} a \rangle  \dot X^{\ell'}_{t-\sqrt{m}\tau} {\rm d}\tau
\, .
\end{split}
\]
We require this dissipation term to be small, so that the coupling to the heat bath yields a small perturbation
of the dynamics for $X$ and $P$. If $\bar V''$ and $\nabla a$ are of order one,
the mass $m$ needs to be small, or if $\nabla a$ is small we can have $\bar V''$ large.
The case with small $m$ is studied in this section and the case with large $\bar V''$ is in
Section \ref{chi_sec}.

A small mass 
also requires the integrand to decay as $\tau\to\infty$.
We will use Fourier analysis to study the decay of the kernel
\[
K_n^{\ell\ell'}:= \langle\bar V'' \cos(\tau \bar V''^{1/2}) \partial_{X^{\ell'}}a , \partial_{X^{\ell}} a\rangle \, ,
\]
by writing  $\lim_{n\to\infty} K_n$ as an integral, which is the next step in the analysis.

 The kernel $K_n$ 
is based on the equilibrium heat bath position derivative $\partial_{X^\ell}a$ and this derivative
is determined by the
derivative of the force on bath particle $x_j$, with respect to $X^\ell$, by \eqref{star} as
\[
F'_{\ell,j}:= -\partial_{X^\ell}\partial_{x_j}\bar V(x,X)
=(\bar V''\partial_{X^\ell} a)_j\, ,
\]
which, for fixed $\ell$ and $j$, by assumptions \eqref{3} and \eqref{Da_const} is independent of $x$ and $X$.
We assume that $\bar V$ is constructed so that this force derivative is localized in the sense
\begin{equation}\label{F_cond}
\begin{split}
&\sum_{j} |F'_{\ell,j}|(1+|j-\hat \ell|^2)=\mathcal O(1)\, , \quad \mbox{ as $n\to\infty$, for $\hat\ell:={\rm argmin}_{j}|X_\ell-x_j|$} \\
&\mbox{ $\lim_{n\to\infty}F'_{\ell,j}$ exists.}
\end{split}
\end{equation}
%

Let $\nu_k(j):=\nu_k'(j)/\max_j|\nu_k'(j)|$ be the set of orthogonal eigenvectors of $\bar V''$
that are normalized to one in the maximum norm.The eigenvalue representation of $\bar V''$ in \eqref{egenvekt}
and the definition
\[
\omega_k^2={\mu_k}
\]
yields
\begin{equation}\label{f_def}
\begin{split}
\beta_{\ell,k} &:= \sum_j  F'_{\ell,j} \nu_k^*(j)=\langle \nu_k,F'_{\ell\cdot}\rangle\, ,\\
 F'_{\ell,j}&=(\bar V'' \partial_{X^\ell} a)_j  = \sum_k \frac{\beta_{\ell,k}}{\|\nu_k\|^2_2} \nu_k(j)\, ,\\
\partial_{X^\ell} a_j &=\sum_k \frac{\beta_{\ell,k}}{\omega_k^2\,  \|\nu_k\|^2_2} \nu_k(j)\, ,\\
\end{split}
\end{equation}
%
%
where $\|\nu_k\|_2:=(\sum_j|\nu_k(j)|^2)^{1/2}$. Consequently we obtain
\begin{equation}\label{K_n_res}
\begin{split}
K_n^{\ell\ell'}(\tau) &= \langle \bar V'' \cos(\tau \bar V''^{1/2}) \partial_{X^\ell}a, \partial_{X^{\ell'}}a\rangle\\
&= \sum_{k} \cos(\tau\omega_k) \frac{\beta_{\ell,k}^* \beta_{\ell',k}}{\omega_k^2\, \|\nu_k\|_2^2}\, .\\
\end{split}
\end{equation}

If $n$ is finite and we make a tiny perturbation of $\bar V''$ so that all $\omega_k$ become rational, the function $K_n$ will be periodic in $\tau$
and consequently it will not decay for large $\tau$. To obtain a decaying kernel we will therefore consider a heat bath with infinite number of particles, $n=\infty$. The next step is consequently to study the limit of $K_n$, as $n\to\infty$, which requires a more precise formulation of the heat bath.

\subsubsection{A precise heat bath}
We use the periodic lattice \[
E_{\bar n}:=
\{-\bar n/2, -\bar n/2+1, \ldots, \bar n/2-1\}^3\subset \rset^3\]
in dimension three to form the equilibrium positions for the heat bath. The position deviation from the equilibrium, namely  $\bar{x}_j:=x_j-a_j(X)\in \rset$ for each particle $j\in E_{\bar n}$, 
then determines the potential by nearest neighbour interaction in the lattice 
\begin{equation}\label{lattice_def}
\begin{split}
\langle \bar V'' \bar{x}, \bar{x}\rangle &= c^2\sum_{i=1}^3 \sum_{j\in E_{\bar n}}
|\bar{x}_{j+e_i}-\bar{x}_j|^2 
+\eta_{n}^2\sum_{i=1}^3 \sum_{j\in E_{\bar n}}
|\bar{x}_{j+e_i}|^2\\
&= c^2\sum_{i=1}^3 \sum_{j\in E_{\bar n}}(-\bar{x}_{j+e_i}+2\bar{x}_{j}-\bar{x}_{j-e_i}) \bar{x}_{j}
+\eta_{n}^2\sum_{i=1}^3 \sum_{j\in E_{\bar n}}
|\bar{x}_{j+e_i}|^2\\
\end{split}
\end{equation}
with periodic boundary conditions $\bar{x}_{j+\bar ne_i}=\bar{x}_j$,
where $j=(j_1,j_2,j_3), \ e_1=(1,0,0), e_2=(0,1,0), e_3=(0,0,1)$,
$ \bar n^3=n$ and $c$ is a positive constant independent of $n$. The small positive constant $\eta_n$ is introduced to make the potential strictly convex for finite $n$, while it vanishes asymptotically and satisfies
\begin{equation}\label{xi_n}
n^{1/2}\eta_n\to \infty
\mbox{ and } \eta_n\to 0^+\  \mbox{ as } n\to\infty\, .
\end{equation}
We will see in \eqref{f_kappa} that the zero limit in \eqref{xi_n} 
is needed in our model to obtain non zero friction matrices $\kappa$.
The lower bound in \eqref{xi_n} implies by \eqref{f_def} that $\|\partial_{X^\ell}a\|_{\ell^\infty}$ remains bounded, provided $\|\beta_{\ell,\cdot}\|_{\ell^1}$ is bounded. 
We note that
$\bar V''$ consists of the standard  finite difference matrix with mesh size one related to the Laplacian in $\mathbb R^3$ and a small positive definite perturbation $\eta_n^2\, {\rm I}$, where ${\rm I}$ is the corresponding identity matrix. The minimum of the potential is obtained for the position deviations $\bar{x}_j=0, \ j\in E_{\bar n}$, which we may view as the $n$ heat bath particles located on the $n$ different lattice points in $E_{\bar n}$. This heat bath model can be extended to positions $\bar{x}_j\in\rset^3$, see Remark \ref{heatbath_extend}.

In each coordinate direction the discrete Laplacian is a circulant matrix so that the eigenvectors and eigenvalues can be written
\begin{equation}\label{hb_model}
\begin{split}
\nu_k(j) &= e^{2\pi {\rm i} j\cdot k/\bar n}\quad,\,  j_i=-\bar n/2,\ldots, \bar n/2-1,\mbox{ and } \ k_i=-\bar n/2, \ldots, \bar n/2-1\, ,\\
\omega_k^2&=c^2\sum_{i=1}^3 2\big(1-\cos(\frac{2\pi k_i}{\bar n})\big)+ \eta_n^2\, ,\\
\end{split}
\end{equation}
and $\|\nu_k\|_2^2=\sum_{i=1}^3 \sum_{j_i=-\bar n/2}^{\bar n/2-1} |\nu_k(j)|^2=\bar n^3$. Let $\mathbf r:=(\frac{k_1}{\bar n}, \frac{k_2}{\bar n},\frac{k_3}{\bar n})$ and $r:=\sqrt{k_1^2+k_2^2+k_3^2}/\bar n$. 
We have for $k=\bar n\mathbf r$ by \eqref{xi_n}
\[
\begin{split}
\nu_k(j) &= e^{2\pi {\rm i} j\cdot \mathbf r}=:\nu(\mathbf r,j)\, ,\\
\lim_{n\to\infty}\omega_k^2&=c^2\sum_{i=1}^3 2\big(1-\cos(2\pi r_i)\big)=: \big(\omega(\mathbf r)\big)^2\, ,\\
\end{split}
\]
and we assume that the derivative of the force
has a limit 
\begin{equation}\label{f_def2}
\lim_{n\to\infty} F'_{\ell,j}=\bar F_{\ell,j}\, ,
\end{equation}
which implies 
\begin{equation}\label{beta_00}
\begin{split}
\lim_{{\tiny \begin{array}{c}
n\to\infty\\
\frac{k}{\bar n}=\mathbf r
\end{array}}} \beta_{\ell,k} &=\sum_{j\in E_\infty} \bar F_{\ell,j}\nu^*(\mathbf r,j)=:\beta_{\ell}(\mathbf r)\, ,\\
\beta_\ell(0) &= \sum_{j\in E_\infty} \bar F_{\ell,j}\, ,
\end{split}\end{equation}
that is,  the function 
$\beta_\ell:[-1/2,1/2]^3\to \rset$ has the Fourier  coefficients $\bar F_{\ell,j}$. The value $\beta_\ell(0)$ will be used in \eqref{f_kappa} to determine the friction matrix $\kappa$.

\begin{remark}\label{heatbath_extend}
The heat bath model \eqref{lattice_def} can be extended to have $\bar x_j=x_j-a_j(X)\in\rset^3$, with $\bar x_j=({\bar x_j^1},{\bar x_j^2},{\bar x_j^3})$. In the case 
\[
\langle\bar V'' \bar x, \bar x\rangle
=\sum_{k=1}^3\sum_{i=1}^3 \sum_{j\in E_{\bar n}}c_k^2(-{\bar x_{j+e_i}^{k}}+2{\bar x_{j}^k}-{\bar x_{j-e_i}^{k}})  {\bar x_{j}^k}
+\eta_{n}^2\sum_{i=1}^3 \sum_{j\in E_{\bar n}}
|\bar x_{j+e_i}|^2\]
we still have the same eigenvalues if $c_k^2=c^2$, so the model does not change in principle. If on the other hand the constants $c_k$ are different for the different components of $\bar x_j$ the spectrum may change and we obtain a different heat bath model.
\end{remark}
\subsubsection{The limit friction matrix}
We can  take the limit as $\bar n\to \infty$ in \eqref{K_n_res}, while $\mathbf r=k/\bar n$ is constant,
to obtain an integral
\begin{equation}\label{K_infty_bdd}
\begin{split}
\lim_{n\to\infty} K_n^{\ell\ell'}(\tau)&=
\lim_{n\to\infty} \sum_{k} \cos(\tau\omega_k) \frac{\beta_{\ell,k}^* \beta_{\ell',k}}{\omega_k^2\|\nu_k\|_2^2}\\
&= \int_{[-\frac{1}{2},\frac{1}{2}]^3} \cos\big(\tau\omega(\mathbf r)\big) 
\frac{\beta_{\ell}^*(\mathbf r) \beta_{\ell'}(\mathbf r)}{
\big(\omega(\mathbf r)\big)^2}{\rm d}r_1 {\rm d}r_2{\rm d}r_3\\
&=:K_\infty^{\ell\ell'}(\tau)\, .
\end{split}
\end{equation}
The change of variables $\omega_i=2^{1/2} c \sqrt{1-\cos(2\pi r_i)}\, {\rm sgn}(r_i)$ yields, with the spherical coordinate
$\boldsymbol\omega=(\cos\alpha\cos\theta,\sin\alpha \cos\theta, \sin\theta)\omega$,
\[
\begin{split}
&\lim_{n\to\infty} \sum_{k} \cos(\tau\omega_k) \frac{\beta_{\ell,k}^* \beta_{\ell',k}}{\omega_k^2\|\nu_k\|_2^2}\\
&= \int_{{\boldsymbol\omega}([-\frac{1}{2},\frac{1}{2}]^3)} \cos(\tau\omega) 
\beta_{\ell}^*(\mathbf r(\boldsymbol\omega)) \beta_{\ell'}(\mathbf r(\boldsymbol\omega))\prod_{i=1}^3 
\Big(\frac{\sqrt{1-\cos\big(2\pi r_i(\omega_i)\big)}}{2^{1/2}\pi c\sin\big(2\pi r_i(\omega_i)\big)}\Big)
\frac{{\rm d}\boldsymbol\omega}{
|\omega|^2}\\
&=\int_{{\boldsymbol\omega}([-\frac{1}{2},\frac{1}{2}]^3)} \cos(\tau\omega) 
f(\boldsymbol\omega,\ell,\ell')\frac{{\rm d}\boldsymbol\omega}{|\omega|^2}\\
&=\int_{0}^{\pi}\int_0^{2\pi}\int_0^\infty
\cos(\tau\omega) 
f(\boldsymbol\omega,\ell,\ell')
\sin\theta{\rm d}\omega {\rm d}\alpha  {\rm d}\theta\\
\end{split}
\]
where
\begin{equation}\label{eq:fOmegaDef}
f(\boldsymbol{\omega},\ell,\ell') := 
\begin{cases}
\beta_{\ell'}(\boldsymbol\omega) \beta_{\ell}^*(\boldsymbol\omega)\prod_{i=1}^3 
\pi^{-1} \big(4c^2 - {\omega_i^2}\big)^{-1/2}\, ,
& -2c<\omega_i<2c\, ,\\
0\, , & \mbox{ otherwise}\, .\\
\end{cases}
\end{equation}
with $\beta(\boldsymbol \omega):= \beta(r(\boldsymbol\omega))$.
The term $(4c^2 - \omega_i^2)^{-1/2}$ is unbounded (but integrable)
at the boundary where $\omega_i=\pm 2c$.
For the purpose of simplifying later proofs we will 
assume that $\beta_\ell(\boldsymbol \omega)$ is two times differentiable 
and that it vanishes at the boundary of ${\boldsymbol\omega}([-\frac{1}{2},\frac{1}{2}]^3) = [-2c,2c]^3$ 
as follows:
\begin{equation}\label{beta_assum}
\beta(\boldsymbol \omega) =0 \mbox{ for  }  \boldsymbol \omega \in [-2c, 2c]^3 \mbox{ satisfying } 
\omega = |\boldsymbol \omega|>c \, .
\end{equation}
We also note that the constant $c$ and the density of states are related by
\begin{equation}\label{dof}
|\frac{\Rd \mathbf r(0)}{\Rd \boldsymbol{\omega}}|=
\prod_{i=1}^3 \frac{\Rd r_i}{\Rd \omega_i}\big|_{\omega_i=0}
= \prod_{i=1}^3 \pi^{-1} \big(4c^2 - \omega_i^2\big)^{-1/2} \Big|_{\omega_i=0}
= (2 c \pi)^{-3} \, .
\end{equation}

We are now ready to formulate the limit as $n\to \infty$ in the friction term based on the constant $N\times N$ 
friction matrix $\kappa$.
\begin{lemma}\label{fric_lem} 
Let 
\[
\kappa_{\ell \ell'}:=
\frac{1}{4\pi c^3}\big(\sum_{j\in E_\infty}  \bar F_{\ell j}\big)
\big(\sum_{j\in E_\infty} \bar F_{\ell' j}\big)
\]
and assume that \eqref{2},\eqref{3}, \eqref{phi0}, \eqref{Da_const}, \eqref{F_cond}, \eqref{lattice_def}, \eqref{beta_assum} hold and for each $t>0$ 
there is a constant $C$ such that 
\[
\begin{split}
\sup_{0\le s\le t}(\mathbb E[|\dot P_s|^2])^{1/2} 
+\sup_{0\le s\le t}(\mathbb E[|\dot X_s|^2])^{1/2}\le C\, ,
\end{split}
\]
then for any  function $h\in L^\infty(\rset^{2N})$ 
\begin{equation}\label{convRateH}
\begin{split}
&\lim_{n\to\infty}
\mathbb E[h(X_t,P_t)\int_0^t \langle \bar V'' \cos\big(\frac{(t-s)\sqrt{\bar V''}}{\sqrt{m}}\big) 
\dot{a}(X_s), \partial_{X^\ell} a\rangle {\rm d} s]\\
&=\mathbb E[m^{1/2}h(X_t,P_t)\kappa \dot X_t] + \mathcal O(m\log m^{-1})\, .
\end{split}
\end{equation}
\end{lemma}

\begin{remark}[Dirchlet boundary condition]If we replace the periodic boundary conditions in the heat bath model \eqref{hb_model} 
with homogenous Dirichlet conditions, 
we have instead 
\[
\begin{split}
\nu_k(j) &= \prod_{i=1}^3 \sin(\frac{\pi k_i}{2} + \frac{\pi j_ik_i}{\bar n+1}) \quad,\,  j_i=-\bar n/2,\ldots, \bar n/2-1,\mbox{ and } \ k_i=1,\ldots, \bar n\, ,\\
\omega_k^2&=c^2\sum_{i=1}^3 2\big(1-\cos(\frac{\pi k_i}{\bar n+1})\big) + \eta_n^2\, ,\\
&\lim_{\bar n\to\infty}\frac{\bar n^3}{\|\nu_k\|_2^2}=8\, .\\
\end{split}
\]
We can write the eigenvectors as functions of $k_i/(\bar n+1)$ since
\[
\begin{split}
\sin(\frac{\pi k_i}{2} + \frac{\pi j_ik_i}{\bar n+1}) = \left\{
\begin{array}{cc}
\sin(\frac{\pi j_ik_i}{\bar n+1}) & \mbox{ if }\mod(k_i,4)=0\, ,\\
\cos(\frac{\pi j_ik_i}{\bar n+1}) & \mbox{ if }\mod(k_i,4)=1\, ,\\
-\sin(\frac{\pi j_ik_i}{\bar n+1}) & \mbox{ if }\mod(k_i,4)=2\, ,\\
-\cos(\frac{\pi j_ik_i}{\bar n+1}) & \mbox{ if }\mod(k_i,4)=3\, ,\\
\end{array}\right.
\end{split}
\]
and split the sum over $k_i$ into one sum over odd $k_i$, where the eigenvector is based on cosine functions, 
and one sum over even $k_i$, where the eigenvector is based on sine functions.
With these changes, the derivation of $\kappa$ follows as in the case with periodic boundary conditions.
 \end{remark}

\begin{proof}[Proof of Lemma \ref{fric_lem}]
To study the decay  as $\tau\to \infty$ of the kernel $K_\infty$, we use \eqref{beta_assum} and the shorthand $f(\boldsymbol{\omega}):=
f(\boldsymbol{\omega},\cdot, \cdot)$ and integrate by parts
\[
\begin{split}
K_\infty(\tau)&=\int_{\mathbb R^3} \cos(\tau\omega) f(\boldsymbol{\omega}) \frac{{\rm d}\boldsymbol{\omega}}{\omega^2}\\
&= \int_{0}^{\pi}\int_0^{2\pi}\int_0^\infty
\cos(\tau\omega) f(\boldsymbol{\omega}) 
\sin\theta{\rm d}\omega {\rm d}\alpha  {\rm d}\theta\\
&=-\int_{0}^{\pi}\int_0^{2\pi}\int_0^\infty
\frac{\sin(\tau\omega)}{\tau} \partial_\omega f(\boldsymbol{\omega}) 
\sin\theta{\rm d}\omega {\rm d}\alpha  {\rm d}\theta\\
&=-\int_{0}^{\pi}\int_0^{2\pi}
\frac{1}{\tau^2} \partial_\omega f(\boldsymbol{\omega})\big|_{\omega=0} 
\sin\theta  {\rm d}\alpha  {\rm d}\theta\\
&\quad -\int_{0}^{\pi}\int_0^{2\pi}\int_0^\infty
\frac{\cos(\tau\omega)}{\tau^2} \partial^2_\omega f(\boldsymbol{\omega}) 
\sin\theta{\rm d}\omega {\rm d}\alpha  {\rm d}\theta\, .\
\end{split}
\]
Since  $f\big(\boldsymbol\omega \big)$ has its support in $\omega =|\boldsymbol \omega|\le c$, cf.~\eqref{beta_assum}, 
and the second derivative $\partial^2_\omega f(\boldsymbol{\omega})$ is bounded, we obtain
\begin{equation}\label{tau_est}
\|K_\infty(\tau)\|=
\|
\int_{\mathbb R^3} \cos(\tau\omega) f(\boldsymbol{\omega}) \frac{{\rm d}\boldsymbol{\omega}}{\omega^2}\|
=\mathcal O\big((1+\tau)^{-2}\big)\, .
\end{equation}

For a given $t>0$ and bounded function $h:\rset^N\times\rset^N\to\rset$ we next study the expectation of 
\[
\lim_{n\to \infty} h(X_t,P_t) \int_0^{t/\sqrt{m}} \langle \bar V'' \cos(\tau \bar V''^{1/2}) \dot a(X_{t-\sqrt{m}\tau}), \partial_{X^\ell}a\rangle
{\rm d} \tau.
\]
In the case $t/\sqrt m>1$, we split the integral with $1<\tau_*<t/\sqrt m$
\begin{equation}\label{4}
\begin{split}
& h(X_t,P_t) \int_0^{t/\sqrt{m}} \langle \bar V'' \cos(\tau \bar V''^{1/2}) \dot a(X_{t-\sqrt{m}\tau}), \partial_{X^\ell}a\rangle
{\rm d} \tau\\
&= h(X_t,P_t)\int_0^{\tau_*}\langle \bar V'' \cos(\tau \bar V''^{1/2}) \dot a(X_{t-\sqrt{m}\tau}), \partial_{X^\ell} a\rangle {\rm d} \tau\\
&\quad + h(X_t,P_t) \int_{\tau_*}^{t/\sqrt{m}}\langle\bar V'' \cos(\tau \bar V''^{1/2}) \dot a(X_{t-\sqrt{m}\tau}), \partial_{X^\ell} a\rangle {\rm d} \tau. 
\end{split}
\end{equation}
The magnitude of the second integral is bounded in expectation using \eqref{K_n_res}, \eqref{K_infty_bdd} and \eqref{tau_est}:
\[
\begin{split}
&\mathbb E[|h(X_t,P_t) \int_{\tau_*}^{t/\sqrt{m}}\langle \bar V'' \cos(\tau \bar V''^{1/2}) \dot a(X_{t-\sqrt{m}\tau}),\partial_{X^\ell} a\rangle {\rm d} \tau|]\\
&\le | \int_{\tau_*}^{t/\sqrt{m}}\frac{C \mathbb E[|h(X_t,P_t)||\dot X_{t-m^{1/2}\tau}|]}{\tau^2}{\rm d} \tau\\
&\le 
\frac{C\sup_{0\le s\le t}(\mathbb E[|\dot X_s|^2])^{1/2}(\mathbb E[|h(X_t,P_t)|^2])^{1/2}}{\tau_*}\, .
\end{split}
\]
%

The expectation 
and the limit $n\to \infty$ of the first integral in the right hand side of~\eqref{4}
can be written as
\begin{multline*}
\mathbb E[h(X_t,P_t)\int_0^{\tau_*} 
\int_{0}^{\pi}\int_0^{2\pi}\int_0^\infty
\cos(\tau\omega) 
f(\boldsymbol\omega)
\sin\theta{\rm d}\omega {\rm d}\alpha  {\rm d}\theta\,  \dot X_{t}\, {\rm d}\tau]\\
- 
\mathbb E[h(X_t,P_t)\int_0^{\tau_*} 
\int_{0}^{\pi}\int_0^{2\pi}\int_0^\infty
\cos(\tau\omega) 
f(\boldsymbol\omega)
\sin\theta{\rm d}\omega {\rm d}\alpha  {\rm d}\theta\,  (\dot X_{t} -\dot X_{t - \tau m^{1/2}})\, {\rm d}\tau],
\end{multline*}
where, using~\eqref{tau_est},
\[
\begin{split}
&\mathbb E[h(X_t,P_t)\int_0^{\tau_*} 
\int_{0}^{\pi}\int_0^{2\pi}\int_0^\infty
\cos(\tau\omega) 
f(\boldsymbol\omega)
\sin\theta{\rm d}\omega {\rm d}\alpha  {\rm d}\theta\,  (\dot X_{t} -\dot X_{t - \tau m^{1/2}})\, {\rm d}\tau]\\
&=\mathbb E[h(X_t,P_t)\int_0^{\tau_*} 
\int_{0}^{\pi}\int_0^{2\pi}\int_0^\infty
\cos(\tau\omega) 
f(\boldsymbol\omega)
\sin\theta{\rm d}\omega {\rm d}\alpha  {\rm d}\theta\, \int^0_{-\tau m^{1/2}} \ddot X_{t+s}{\rm d}s
\, {\rm d}\tau]\\
&\le Cm^{1/2} \int_0^{\tau_*} \frac{\tau}{1+\tau^2}{\rm d}\tau
\sup_{0\le s\le t}(\mathbb E[|\ddot X_s|^2])^{1/2}(\mathbb E[|h(X_t,P_t)|^2])^{1/2}\\
&\le  Cm^{1/2}\log \tau_*
\sup_{0\le s\le t}(\mathbb E[|\ddot X_s|^2])^{1/2}(\mathbb E[|h(X_t,P_t)|^2])^{1/2}
\, .
%
%
\end{split}
\] 
We prove in Section \ref{system_dyn} that the expected value $\sup_{0\le s\le t}(\mathbb E[|\dot X_s|^2+|\ddot X_s|^2])$
 is bounded, and by assumption $\mathbb E[|h(X_t,P_t)|^2]$
is bounded. Therefore
we have  the error estimate
\begin{equation}\label{6}
\begin{split}
&|\mathbb E\big[h(X_t,P_t)\big( \sum_{\ell'} \int_0^{t/m^{1/2}}K_\infty^{\ell,\ell'}(\tau)
\dot X_{t-\tau m^{1/2}}^{\ell'} {\rm d}\tau
-\sum_{\ell'} \int_0^{\infty}K_\infty^{\ell,\ell'}(\tau)
\dot X_{t}^{\ell'} {\rm d}\tau\big)\big]|\\
& =
\frac{C}{\tau_*} + C m^{1/2}\log \tau_*\, ,
\end{split}
\end{equation}
so that with $\tau_*=m^{-\frac{1}{2}}$ the error in \eqref{convRateH}
is bounded by 
\begin{equation}\label{77}
m^{1/2}\min_{\tau_*}(\frac{C}{\tau_*} + C m^{1/2}\log\tau_*)=\mathcal O(m\log m^{-1})\, .
\end{equation}
The Fourier transform of $f(\boldsymbol{\omega})$
with respect to $\omega$ is integrable and since also $f(\boldsymbol{\omega})$ is continuous,
we have by the Fourier inversion property and \eqref{beta_00}
\begin{equation}\label{f_kappa}
\begin{split}
&\lim_{m\to 0^+}\lim_{n\to\infty}
m^{-1/2}\int_0^t \langle \bar V'' \cos\big(\frac{(t-s)\sqrt{\bar V''}}{\sqrt{m}}\big) \dot a(X_s), \partial_{X^\ell} a\rangle {\rm d} s\\
&=
\sum_{\ell'}\int_0^\infty
\int_{\mathbb R^3} \cos(\tau\omega) f(\boldsymbol{\omega},\ell,\ell') 
\frac{{\rm d}\boldsymbol{\omega}}{\omega^2}{\rm d}\tau \dot X_t^{\ell'}\\
&=
\sum_{\ell'}\int_0^\infty \int_{0}^\infty \int_0^{2\pi}\int_{0}^{\pi}
 \cos(\tau\omega) f(\boldsymbol{\omega},\ell,\ell')   \sin\theta{\rm d}\theta
{\rm d}\alpha {\rm d}\omega {\rm d}\tau\, \dot X_t^{\ell'}\\
&=\frac{\pi}{2}\sum_{\ell'} \lim_{\omega\to 0} \int_0^{2\pi}\int_{0}^{\pi} f(\boldsymbol\omega,\ell,\ell')   \sin\theta{\rm d}\theta{\rm d}\alpha\, \dot X_t^{\ell'}\\
&=2\pi^2 \sum_{\ell'}f(0,\ell,\ell')\dot X_t^{\ell'}\\
&=\frac{1}{4\pi c^3}\sum_{\ell'}\big(\sum_{j\in E_\infty}  \bar F_{\ell j}\big)
\big(\sum_{j\in E_\infty} \bar F_{\ell' j}\big)\dot X_t^{\ell'}\\
&=\kappa\dot X_t\, .
\end{split}
\end{equation}
That is, the friction  matrix, $\kappa$, in the Langevin equation is determined by 
the $\partial_{X^\ell}$-derivative 
of the sum of forces
on all bath particles and we have proved Lemma \ref{fric_lem}.
\end{proof}

 \begin{remark}[Vanishing friction]\label{rem_vanish}  
 
 If the heat bath model is modified to have nearest neighbor  interactions in a lattice in dimension $d$,
we obtain as in \eqref{f_kappa}
  \[
  \kappa_{\ell \ell'}= S_d \lim_{\omega\to 0} f(\boldsymbol{\omega},\ell,\ell')\omega^{d-1}/\omega^2
  =  S_d \lim_{\omega\to 0} \beta_\ell(\boldsymbol{\omega}) \beta_{\ell'}^*(\boldsymbol{\omega}) \omega^{d-3}
  \]
 where $S_d$ is a positive constant related to the dimension $d$.
 We see that under the assumption $\beta_{\ell}(0^+) \beta_{\ell'}(0^+) >0$,
 it is only in dimension $d=3$ that this heat bath generates a positive definite friction matrix $\kappa$.
 For $d<3$ we obtain $\kappa=\infty$ and for $d>3$ we have $\kappa=0$.
%

 If we change the heat bath potential energy to
 be based on any  circulant matrix in each dimension, we have the requirement 
 \begin{equation}\label{req2}
 \lim_{\omega\to 0} \frac{r^2(\boldsymbol{\omega})\frac{{\rm d}r}{{\rm d}\omega}}{\omega^2}=\mbox{constant}
 \end{equation}
 to obtain a positive definite friction matrix $\kappa$, in dimension $d=3$.
 This limit for the eigenvalues $\omega^2$ implies that  $\bar V''$  becomes a difference quotient 
 approximation of the Laplacian with mesh size one.
 We conclude that \eqref{req2} leads to a choice of $\bar V''$ in \eqref{lattice_def} that is another 
 discretization of the Laplacian or
 discretizations that tends to the Laplacian as $n\to \infty$.
 \end{remark}

\subsection{The fluctuation term}\label{sec_fluc} The initial distribution  of 
$\varphi_0=\sum_k(\gamma_k^r+{\rm i}\gamma_k^i )\nu_k'$,
determined by the Gibbs distribution of the Hamiltonian system in \eqref{phi0} and \eqref{gamma0},
shows that all $\gamma_k^r$ and $\gamma_k^i$ are independent and
normal distributed with mean zero and variance $T/\omega^2_k$.
This initial data 
%
provides the fluctuation term in \eqref{1.5}, namely 
\begin{equation}\label{gamma_W}
\begin{split}
\zeta_t^\ell&= {\rm Re}\langle e^{-{\rm i}t\sqrt{\bar V''}/\sqrt{m}}\varphi_0, \bar V'' \partial_{X^\ell}a\rangle\\
&=\sum_k{\rm Re}( e^{{\rm i}t\omega_k/\sqrt{m}}\frac{\gamma_k^*}{\bar n^{3/2}} 
\langle \nu_k, \bar V'' \partial_{X^\ell}a\rangle)\\
& =\sum_k {\rm Re}( e^{{\rm i}t\omega_k/\sqrt{m}}\frac{\gamma_k^*}{ n^{1/2}}\beta_{k,\ell}) \, ,
\end{split}
\end{equation}
which has the special property that 
its covariance  satisfies 
\begin{equation}\label{fluc_diss}
\begin{split}
\mathbb E[\zeta_t^\ell\zeta_s^{\ell'}] 
&= \mathbb E[{\rm Re}\langle {e^{-{\rm i}t\sqrt{\bar V''}/\sqrt{m}}}\varphi_0, \bar V'' \partial_{X^\ell} a\rangle
\ {\rm Re}\langle e^{-{\rm i}s\sqrt{\bar V''}/\sqrt{m}}\varphi_0,
 \bar V'' \partial_{X^{\ell'}} a\rangle ]\\
&=T\langle \cos\big((t-s)\sqrt{\bar V''}/\sqrt{m})\bar V'' \partial_{X^\ell} a ,\partial_{X^{\ell'}}a\rangle 
\end{split}
\end{equation}
that is, the covariance is $T$ times the friction integral kernel, as observed in \cite{zwanzig_1973}.
We repeat a proof for our setting here, where we also show that the limit of $\zeta_t$ as $n\to\infty$ is well defined.
\begin{proof}[Proof of \eqref{fluc_diss}] Let
$e^{-{\rm i}t\sqrt{\bar V''}/\sqrt{m}}=:S_{t0}$ and write $e^{-{\rm i}t\sqrt{\bar V''}/\sqrt{m}}\varphi_0=
S_{ts}S_{s0}\varphi_0$. Since the operator $S_{s0}$ is unitary,
we have $S_{s0}\varphi_0=\sum_k\gamma_k'\nu_k'$
where $\{\nu_k'\}$ is the set of normalized real valued eigenvectors of $\bar V''$,
defined in \eqref{phi0}. The random coefficients 
$\gamma_k'=\gamma_{k,r}'+{\rm i}\gamma_{k,i}'$ are based on  $\gamma_{k,r}'$ and $\gamma_{k,i}'$
which are independent normal  distributioned with mean zero and variance $T/\omega_k^2$, for $k=(k_1,k_2,k_3)$ and $k_i=1,\ldots, \bar n$.
Use
the orthonormal real valued basis $\{\nu_k'\}$  to obtain
\[
\begin{split}
\mathbb E[\zeta_t^\ell\zeta_s^{\ell'}] &= \mathbb E\Big[\sum_{k,k'}
\big\langle {\rm Re}\big( S_{t0}(\gamma_k^r+{\rm i}\gamma_k^i)\nu_k'\big), \bar V''\partial_{X_\ell} a\big\rangle\ 
\big\langle {\rm Re}\big(S_{s0}
(\gamma_{k'}^r+{\rm i}\gamma_{k'}^i)
\nu_{k'}'\big), \bar V''\partial_{X_{\ell'}} a\big\rangle \Big]\\
 &=\mathbb E\Big[\sum_{k,k'}
\big\langle
{\rm Re}\big( S_{ts}(\gamma_{k,r}'+{\rm i}\gamma_{k,i}')\nu_k'\big), \bar V''\partial_{X_\ell} a\big\rangle\ 
\big\langle {\rm Re}\big(
(\gamma_{k',r}'+{\rm i}\gamma_{k',i}')
\nu_{k'}'\big),\bar V''\partial_{X_{\ell'}} a\big\rangle\Big] \\
&=\mathbb E\Big[\sum_{k,k'}
\big\langle
\omega_k^2 \Big(\cos\big((t-s)\omega_km^{-1/2}\big)\gamma_{k,r}' -\sin\big((t-s)\omega_km^{-1/2}\big)\gamma_{k,i}'\Big)\nu_k', \partial_{X_\ell} a\big\rangle\times\\
&\qquad\times\big\langle \gamma_{k',r}'
\nu_{k'}',\bar V''\partial_{X_{\ell'}} a\big\rangle\Big] \\
&=\sum_{k,k'} \Big(
\big\langle
\omega_k^2 \Big(\cos\big((t-s)\omega_km^{-1/2}\big)\mathbb E[\gamma_{k,r}' \gamma_{k',r}'] \\
&\quad -\sin\big((t-s)\omega_km^{-1/2}\big)\mathbb E[\gamma_{k,i}' \gamma_{k',r}']\Big)\nu_k', \partial_{X_\ell} a\big\rangle\big\langle
\nu_{k'}',\bar V''\partial_{X_{\ell'}} a\big\rangle\Big) \\
&=\sum_{k} T
\big\langle\cos\big((t-s)\sqrt{\bar V''}/\sqrt{m}\big) \nu_k',\partial_{X_\ell} a\big\rangle\big\langle \nu_k', \bar V'' \partial_{X_{\ell'}} a\big\rangle\\
&= T\langle
\cos\big((t-s)\sqrt{\bar V''}/\sqrt{m}\big) \partial_{X_{\ell}} a,\bar V'' \partial_{X_{\ell'}} a\rangle\, .\\
\end{split}
\]

We conclude that 
$\zeta$ is a Gaussian process with mean  zero and covariance \eqref{fluc_diss}.
The fluctuation-dissipation property \eqref{fluc_diss}, the limit \eqref{K_infty_bdd}  and the bound \eqref{tau_est} show that
the covariance matrix has a limit as $n\to\infty$:
\begin{equation}\label{cov_k}
\lim_{n\to\infty}\mathbb E[\zeta_t\zeta_s]=T K_\infty\big(\frac{t-s}{\sqrt m}\big)\, .
\end{equation}
Therefore $\zeta$ has a well defined limit, in the $L^2$-space with norm 
$({\int_0^\tau \mathbb E[|\zeta_t|^2]{\rm d} t})^{1/2}$, for any finite time $\tau$  as $n\to\infty$.
\end{proof}
\subsection{The system dynamics}\label{system_dyn}

We can write the system dynamics \eqref{1.2} and  \eqref{1.5} as
\begin{equation}\label{xp_ekv}
\begin{split}
X_t &= X_0 + \int_0^t P_s {\rm d} s\, ,\\
P_t &= P_0 -\int_0^t \nabla \lambda(X_s) {\rm d} s + \int_0^t {\rm Re} \langle \varphi_s, \bar V''\nabla a\rangle
 {\rm d} s\\
&= P_0 -\int_0^t \nabla \lambda(X_s) {\rm d} s - \int_0^t \int_0^s  {\rm Re}\langle
e^{-{\rm i} (s-\sigma)\bar V''m^{-1/2}}
 (P_{\sigma}\cdot\nabla) a , \bar V''\nabla a \rangle{\rm d}\sigma {\rm d} s\\
&\quad +\int_0^t \zeta_s{\rm d}s\, .
\end{split}
\end{equation}
We assume that 
\begin{equation}\label{d2lambda}
\begin{split}
&\mbox{$\lambda$ is three times continuously differentiable and}\\
&\mbox{$\sup_X\|{\dD}^2 \lambda(X)\| +\sup_X\|{\dD}^3 \lambda(X)\|$ is bounded}
\end{split}
\end{equation}
and apply Cauchy's inequality  on \eqref{xp_ekv} to obtain
\begin{equation}\label{XP_bdd}
\begin{split}
\frac{{\rm d}}{{\rm d}t} \int_0^t |X_s|^2+|P_s|^2 {\rm d} s &=|X_t|^2+|P_t|^2\\
& \le 
C\big(1+ (t+t^3)\int_0^t (|X_s|^2+ |P_s|^2) {\rm d} s + t\int_0^t |\zeta_s|^2 {\rm d} s\big)\, ,
\end{split}
\end{equation}
which implies
\[
\frac{{\rm d}}{{\rm d}t}\int_0^t  \mathbb E[|X_s|^2+|P_s|^2] {\rm d} s \le 
C\big(1+ (t+t^3)\int_0^t  \mathbb E[|X_s|^2+ |P_s|^2] {\rm d} s + t\int_0^t  \underbrace{\mathbb E[|\zeta_s|^2]}_{\le T{\rm trace}(K_\infty(0))} {\rm d} s\big)\, .
\]
Here ${\dD}^k \lambda$ denotes the set of all partial derivatives of order $k$ and if $k=1$ or $k=2$ we identify it with the gradient and the Hessian, respectivly. 
Integration yields the Gronwall inequality
\[
\int_0^t  \mathbb E[|X_s|^2+ |P_s|^2] {\rm d} s\le Ce^{Ct^4}\, ,
\]
which by \eqref{XP_bdd}  and the mean square of
\[
\dot P_t=
- \nabla \lambda(X_t)  +   \int_0^t {\rm Re}\langle e^{-{\rm i} (t-s)\bar V''m^{-1/2}}
 (P_{s}\cdot \nabla) a ,\bar V''\nabla a\rangle  {\rm d} s +\zeta_t
\]
establishes
\begin{lemma}\label{expected_values_lemma}
Suppose  that the assumptions in Lemma \ref{fric_lem} and \eqref{d2lambda} hold, then for each $t>0$ there is a constant $C$  such that
\[
\begin{split}
\sup_{0\le s\le t}  \mathbb E[|X_s|^2+ |P_s|^2+|\dot P_s|^2] &\le C(1+  e^{Ct^4})\, .\\
\end{split}
\]
\end{lemma}
\subsection{Approximation by Langevin dynamics}\label{du}
In this section we approximate the Hamiltonian dynamics \eqref{1.2} by the Langevin dynamics \eqref{Ito_langevin}
\begin{equation*} 
\begin{split}
{\rm d}X_L(t) &= P_L(t){\rm d}t\, ,\\
{\rm d}P_L(t) &= -\nabla \lambda\big(X_L(t)\big){\rm d}t - m^{1/2}\kappa P_L(t){\rm d}t  + (2 m^{1/2}\kappa T)^{1/2}\ {\rm d} W(t)\, ,\\
X_L(0) &=X(0)\, ,\\
P_L(0) &=P(0)\, .\\
\end{split}
\end{equation*}
To analyse the approximation we use that, 
\begin{equation}\label{g_smooth}
\mbox{for any infinitely differentiable function $g:\rset^{2N}\to\rset$ with compact support, }
\end{equation}
the expected value
\[
u(z,s)= \mathbb E[g\big(X_L(t_*), P_L(t_*)\big)\ |\ \big(X_L(s),P_L(s)\big)=z]\, , \quad s<t_* \mbox{ and } z\in \rset^{2N},
\]
is  well defined, since the Langevin equation \eqref{Ito_langevin}, for $Z_t:=\big(X_L(t),P_L(t)\big)$, has Lipschitz continuous drift and constant diffusion coefficient.
The assumption \eqref{d2lambda} implies that
the stochastic flows $\frac{\partial Z_t}{\partial Z_s}$ and $\frac{\partial^2 Z_t}{\partial Z_s^2}$
also are well defined, which implies that the function $u$ has bounded and continuous derivatives up to order two  in $z$ and to order one in $t$, see \cite{krylov} and \cite{harrier_jensen},
 and solves the corresponding Kolmogorov equation
\[
\begin{split}
&\partial_s u(X,P,s) + P\cdot \nabla_X u(X,P,s)
- \big(\nabla\lambda(X) + \sqrt{m} \kappa P\big)\cdot \nabla_Pu(X,P,s)\\
&\quad +\sum_{j,k=1}^N \sqrt{m}T\kappa_{jk} \partial_{P_jP_k}u(x,P,s)=0\, ,\quad s<t_*\, ,\\
&u(X,P,t_*) = g(X,P)\, .
\end{split}
\]
We have 
\begin{equation}\label{u_evol}
\begin{split}
&\mathbb E [ g\big(X(t_*),P(t_*)\big)\ |\ X_0,P_0]-
\mathbb E [ g\big(X_L(t_*),P_L(t_*)\big)\ |\ \big(X_L(0),P_L(0)\big)=(X_0,P_0)]\\
&= \mathbb E [u\big(X(t_*),P(t_*),t_*\big)-u\big(X(0),P(0),0\big)\ |\ X_0,P_0]\\
&=\int_0^{t^*} \mathbb E[{\rm d} u\big(X(t),P(t),t\big)\ |\ X_0,P_0]\\
&=\int_0^{t_*} \mathbb E[\partial_t u(X_t,P_t,t) + P_t\cdot \nabla_X u(X_t,P_t,t)
+\dot P_t\cdot \nabla_P u(X_t,P_t,t)\ |\ X_0,P_0]{\rm d} t\\
&=\int_0^{t_*} \mathbb E\big[(\dot P_t+\nabla\lambda(X_t)+\sqrt{m}\kappa P_t)\cdot
\nabla_Pu(X_t,P_t,t)\\
&\qquad - \sqrt{m} T \kappa  : {\dD}^2_{PP} u(X_t,P_t,t) \ |\ X_0,P_0\big]{\rm d} t\\
&= \int_0^{t_*} \mathbb E\big[\big({\rm Re}\langle \varphi_t,\bar V''\nabla a\rangle +\sqrt{m}\kappa P_t\big) \cdot
\nabla_P u(X_t,P_t,t)\\
&\qquad - \sqrt{m} T \kappa: {\dD}^2_{PP} u(X_t,P_t,t) \ |\ X_0,P_0\big] {\rm d} t\, ,\\
\end{split}
\end{equation}
where $\kappa:{\dD}^2_{PP}u := \sum_{\ell,\ell'}\kappa_{\ell,\ell'} \partial_{P_\ell}\partial_{P_{\ell'}}u$.

\begin{lemma}\label{E_estim1}
Suppose that the assumptions in Lemma \ref{expected_values_lemma} and \eqref{g_smooth} hold, then
\begin{equation}\label{a}
\begin{split}
&\lim_{n\to\infty}
\mathbb E[\int_0^t \langle \cos\big(\frac{(t-s)\sqrt{\bar V''}}{\sqrt{m}}\big) 
\dot{a}(X_s),\bar V''\nabla a\cdot\nabla_P u(X_t,P_t)\rangle{\rm d} s\ |\ X_0,P_0]\\
&=m^{1/2}\mathbb E[\kappa \dot X_t \cdot\nabla_P u(X_t,P_t)\ |\ X_0,P_0] + \mathcal O(m\log m^{-1})\, .\\
\end{split}\end{equation}
and
\begin{equation}\label{b}
\begin{split}
& \lim_{n\to\infty}
\mathbb E[\big\langle {\rm Re}(e^{-{\rm i}t V''^{1/2}/\sqrt{m}}\varphi_0) , \bar V''\nabla a\cdot\nabla_P u(X_t,P_t,t)\big\rangle \ |\ X_0,P_0] \\
&= m^{1/2}T\mathbb E[\kappa: {\dD}_{PP} u(X_t,P_t,t)\ |\ X_0,P_0] +
\mathcal O(m\log m^{-1})\, .\\
\end{split}
\end{equation}
\end{lemma}

The limit \eqref{a} is verified in \eqref{6} and
we prove \eqref{b} below.
The lemma and the error estimates \eqref{6} and \eqref{77} inserted in \eqref{u_evol} imply 
\begin{theorem}\label{thm1}Provided the assumtions in Lemma \ref{E_estim1} hold, then
the Langevin dynamics \eqref{Ito_langevin} approximates the Hamiltonian dynamics \eqref{1.2} and \eqref{1.5},
with the error estimate 
\[
\Big|\mathbb E[ g(X_{t_*},P_{t_*})\, |\, X_0,P_0]
-\mathbb E[ g\big(X_L(t_*,P_L(t_*)\big)\, |\, \big(X_L(0),P_L(0)\big)=(X_0,P_0)] \Big|=
\mathcal O(m\log \frac{1}{m})\, , 
\]
where the rank one friction matrix is determined by the force $\bar F$ from \eqref{f_def} and \eqref{f_def2}
\begin{equation}\label{kap1}
\kappa_{\ell \ell'}= \frac{1}{4\pi c^3}\big(\sum_{j\in E_\infty} \bar F_{\ell j}\big)
\big(\sum_{j'\in E_\infty} \bar F_{\ell' j'}\big)
=\frac{1}{4\pi c^3}\sum_{j\in E_\infty} 
\big(\bar V''\partial_{X^{\ell}}a(X)\big)_j
\sum_{j'\in E_\infty} \big(\bar V''\partial_{X^{\ell'}}a(X)\big)_{j'}
\, .
\end{equation}
\end{theorem}


We see also that the alternative  dynamics, with any friction coefficient $\bar\kappa\ge 0$,
\[
\begin{split}
{\rm d}{\bar X}_t &= \bar P_t{\rm d}t\\
{\rm d}{\bar P}_t&= -\nabla\lambda(\bar X_t){\rm d}t -\bar\kappa \bar P_t {\rm d}t
+ (2\bar\kappa T)^{1/2}{\rm d}W_t
\end{split}
\]
approximates \eqref{1.2} and \eqref{1.5} with the larger error 
\begin{equation}\label{larger_err}
\big|\mathbb E[ g(X_{t_*},P_{t_*})\ |\ X_0,P_0] -\mathbb E[g(\bar X_{t_*},\bar P_{t_*})\ |\ \bar X_0=X_0, \bar P_0=P_0]\big|
=\mathcal O\big(\max(m^{1/2},\|\bar\kappa\|)\big)\, ,
\end{equation}
unless $\bar\kappa=m^{1/2}\kappa$.
\begin{proof}[Proof of \eqref{b}.]

We have 
\begin{equation}\label{rhs1}
{\rm Re}\langle\varphi_t , \bar V''\nabla a\rangle= \langle {\rm Re}(\underbrace{e^{-{\rm i}t\bar V''/m^{1/2}}\varphi_0}_{=:\varphi_1(t)}) , \bar V''\nabla a\rangle
-\int_0^t \langle\cos\big(\frac{t-s}{\sqrt{m}}\bar V''^{1/2}\big) \dot a(X_s), \bar V''\nabla a\rangle {\rm d} s
\end{equation}
and  \eqref{a} shows that the second term
cancels  $\sqrt{m}\kappa P_t\cdot
\nabla_P u(X_t,P_t,t)$ to leading order in \eqref{u_evol}. It remains to show that the scalar product of $\nabla_P u(X_t,P_t,t)$
and the first term in the right hand side \eqref{rhs1} 
 cancels the last term in \eqref{u_evol}.

To analyze 
the dependence between $\partial_Pu(X_t,P_t,t)$ and 
${\rm Re}\langle e^{-{\rm i}t\bar V''/m^{1/2}}\varphi_0 , \bar V''\nabla a\rangle$ from the initial data  $\varphi_0$ we will study how
a small perturbation of $\varphi_0$ influences $\partial_Pu(X_t,P_t,t)$.

The proof has three steps: 
\begin{itemize}
\item[(1)] 
to derive a representation of $Z_t=(X_t,P_t)$ in terms of perturbations of the initial data  $\varphi_0=\sum_k\gamma_k\nu'_k$ by removing one term, 
\item[(2)] for a given function $h$ to determine expected values $\mathbb E[\langle e^{-{\rm i}t\bar V''^{1/2}m^{-1/2}}\varphi_0,h(Z_t)\rangle]$  using Step (1), and
\item[(3)] to evaluate \eqref{b}, using the conclusion from Step (2) and the covariance result \eqref{fluc_diss}.
\end{itemize}

{\it Step 1.}
{\it Claim.}
Consider two different  functions: $\varphi_1(t)$ and $\varphi_2(t)=\varphi_1(t)+\epsilon v(t)$,
where $\epsilon\ll 1$, 
then the corresponding solution paths $Z_1=(X_1,P_1)$ and $Z_2=(X_2,P_2)$ of the dynamics 
\eqref{1.2} (which can be written as  
\eqref{1.5}) based on the functions $\varphi_1$ and $\varphi_2$, respectively,  
satisfy
\begin{equation}\label{perturb}
\begin{split}
Z_2(t)&=
Z_1(t)+\sum_{\ell=1}^N \begin{bmatrix}
\int_0^t G_{X_1 P_\ell}(t,s) {\rm Re} \langle
\epsilon v(s) , \bar V''\partial_{X_\ell} a\rangle {\rm d} s\\
\vdots\\
\int_0^t G_{X_N P_\ell}(t,s) {\rm Re} \langle
\epsilon v(s) , \bar V''\partial_{X_\ell} a\rangle {\rm d} s\\
\int_0^t G_{P_1 P_\ell}(t,s) {\rm Re} \langle
\epsilon v(s) , \bar V''\partial_{X_\ell} a\rangle {\rm d} s\\
\vdots\\
\int_0^t G_{P_N P_\ell}(t,s) {\rm Re} \langle
\epsilon v(s) , \bar V''\partial_{X_\ell} a\rangle {\rm d} s\\
\end{bmatrix}
\\
&=
Z_1(t)+ \sum_{\ell=1}^N \int_0^t G_{\cdot P_\ell}(t,s) {\rm Re} \langle
\epsilon v(s) , \bar V''\partial_{X_\ell} a\rangle {\rm d} s
\end{split}
\end{equation}
where 
\[
G(t,s) = 
\begin{bmatrix} G_{XX}(t,s) & G_{XP}(t,s)\\
G_{PX}(t,s) & G_{PP}(t,s)
\end{bmatrix} \in \rset^{2N\times 2N}
\]
solves the linear equation 
\begin{equation}\label{g_ekv22}
\begin{split}
\partial_t G_{X y}(t,s) &=G_{P y}(t,s)\, ,\\
\partial_t G_{P y}(t,s) 
&= -\int_0^1{\dD}^2\lambda\big(\tau X_1(t)+(1-\tau)X_2(t)\big){\rm d}\tau \,
G_{X y}(t,s)\\
&\quad -\int_s^t\langle \bar V''\cos\big(\frac{(t-r)\bar V''^{1/2}}{m^{1/2}}\big)\nabla a\cdot G_{P y}(r,s),\nabla a\rangle {\rm d}r\, ,\ t>s,\ y=X,P,\\
G(s,s) &= 
\left[\begin{array}{cc}
{\rm I} & 0\\
0 & {\rm I}\end{array}
\right] \, .
\end{split}
\end{equation}

{\it Proof of the claim.}
The linearized problem  corresponding to \eqref{1.5} becomes
\begin{equation}\label{linearized}
\begin{split}
    \frac{{\rm d}}{{\rm d}t}{\bar X}_t &= \bar P_t\, ,\\
    \frac{{\rm d}}{{\rm d}t} {\bar P}_t
    &= -\int_0^1{\dD}^2\lambda\big(\tau X_1(t)+(1-\tau)X_2(t)\big){\rm d}\tau \, \bar X_t\\
&\quad -\int_0^t\langle \bar V''\cos\big(\frac{(t-s)\bar V''^{1/2}}{m^{1/2}}\big)\nabla a\cdot \bar P_s,\nabla a\rangle {\rm d}s\\
&\quad +{\rm Re}\langle e^{-{\rm i}t\bar V''/m^{1/2}}\bar\varphi_0 , \bar V''\nabla a\rangle\, ,\ t>0.\\
\end{split}
\end{equation}
Consider $R(t):={\rm Re}\langle e^{-{\rm i}t\bar V''/m^{1/2}}\bar\varphi_0 , \bar V''\nabla a\rangle
= \langle \epsilon v(t),\bar V''\nabla a\rangle$
as a perturbation to the linearized equation \eqref{linearized} with $\bar Z(t):=(\bar X_t,\bar P_t)=Z_2(t)-Z_1(t)$ and 
$\bar R(t):=\big(0, R(t) \big) \in \mathbb{R}^{2N}$
and write \eqref{perturb} as
\begin{equation}\label{Z_ekv}
\bar Z(t)= \int_0^t G(t,s)\bar R(s){\rm d}s\, 
\end{equation}
and \eqref{g_ekv22} in abstract form as
\begin{equation}\label{G_ekv_abs}
\begin{split}
\partial_t G(t,s) &= A(t) G(t,s)
- \int_s^t L(t,v) G(v,s) {\rm d} v\, , \ t>s\, ,\\
G(s,s)&={\rm I}\, .
\end{split}
\end{equation}
Differentiation and change of the order of integration imply by \eqref{Z_ekv} and \eqref{G_ekv_abs}
\[
\begin{split}
\frac{{\rm d}}{{\rm d} t}\bar Z(t)
&= G(t,t) \bar R(t) + \int_0^t\partial_t G(t,s) \bar R(s){\rm d}s\\
&=\bar R(t)+A(t)\int_0^t G(t,s)\bar R(s){\rm d}s
-\int_0^t\int_s^t L(t,v) G(v,s) \bar R(s){\rm d}v{\rm d}s\\
&=\bar R(t)+A(t)\bar Z_t -\int_0^t L(t,v)\int_0^v G(v,s)\bar R(s){\rm d}s{\rm d}v\\
&=\bar R(t)+A(t)\bar Z(t)-\int_0^t L(t,v) \bar Z(v){\rm d}v\, ,
\end{split}
\]
which shows that \eqref{perturb} satisfies the linearized equation \eqref{linearized}, and the claim is proved.

The main reason
 we use assumption  \eqref{Da_const}, namely that $\nabla a(X)$ is constant,
 is to obtain this linearized equation for $G$, which otherwise would include the term 
 ${\rm Re}\langle \varphi, \bar V'' {\dD}^2a\rangle$ that would introduce the fast time scale of $\varphi$ in $G$, which we now avoid. 
 
{\it Step 2.} 
We will use $\varphi_1(t):=\sum_{j}e^{-{\rm i}t\bar V''m^{-1/2}} (\gamma_j^r +{\rm i}\gamma_j^i)\nu_j'$,
based on the orthonormal  eigenvectors $\{\nu_k'\}$ in \eqref{phi0}, and for a given $k$ define
$\varphi_2(t):=\sum_{j\ne k}e^{-{\rm i}t\bar V''m^{-1/2}} (\gamma_j^r +{\rm i}\gamma_j^i)\nu_j'$, that is
$\epsilon v(t)=-e^{-{\rm i}t\omega_km^{-1/2}}
(\gamma_k^r +{\rm i}\gamma_k^i)\nu_k'$, so that $Z_1=(X_1,P_1)$ is the path 
corresponding to $\varphi_1$ and $Z_2=(X_2,P_2)$ corresponds to $\varphi_2$ where $\gamma_k\nu_k'$ is removed from the sum in $\varphi_1$.
For any bounded function $h:\rset^{2N}\to \mathbb C^{n}$ the perturbation property \eqref{perturb}, with $\varphi_0=\sum_k\gamma_k\nu_k'$, implies
\begin{equation}\label{expand}
\begin{split}
&\mathbb E[\langle e^{-{\rm i} t\bar V''^{1/2}m^{-1/2}}\varphi_0, h\big(Z_1(t)\big)\rangle]\\
&=\sum_k\mathbb E[\gamma_k^* \underbrace{\langle 
\nu_k',  e^{-{\rm i} t\bar V''^{1/2}m^{-1/2}}h\big(Z_1(t)\big)
\rangle}_{=:h_k(Z_1(t))}]\\
&=\sum_k\mathbb E\Big[\gamma_k^*\Big(h_k(Z_2)\\
&\quad+\int_0^1\nabla h_k\big(\sigma Z_2+(1-\sigma)Z_1\big){\rm d}\sigma \cdot
\sum_\ell\int_0^t G_{\cdot P_\ell}(t,s){\rm Re}
\langle e^{-{\rm i}s\omega_km^{-1/2}}\gamma_k\nu_k' , \bar V''\partial_{X_\ell} a\rangle {\rm d} s\Big)\Big]\\
&=\sum_k\underbrace{\mathbb E[\gamma_k^*h_k(Z_2)]}_{=\mathbb E[\gamma_k^*]\mathbb E[(h_k(Z_2)]=0} +
\sum_k\mathbb E[\gamma_k^*\int_0^1\nabla h_k\big(\sigma Z_2+(1-\sigma)Z_1\big){\rm d}\sigma\\
&\qquad \cdot
\sum_\ell\int_0^t G_{\cdot P_\ell}(t,s){\rm Re}
 \langle e^{-{\rm i}s\omega_km^{-1/2}} \gamma_k\nu_k' , \bar V''\partial_{X_\ell} a\rangle {\rm d} s]\\
 &=\sum_k\mathbb E[\gamma_k^* \nabla h_k(Z')
\cdot
\sum_\ell\int_0^t G_{\cdot P_\ell}(t,s )
 {\rm Re}\langle e^{-{\rm i}s\omega_km^{-1/2}}\gamma_k\nu_k' , \bar V''\partial_{X_\ell} a\rangle {\rm d} s]\, ,\\
\end{split}
\end{equation}
where $Z'$ is between $Z_1(t)$ and $Z_2(t)$ and satisfies \[
\nabla h\big(Z'(t)\big)=\int_0^1\nabla h\Big( Z_2(t)+\sigma\big(Z_1(t)- Z_2(t)\big)\Big){\rm d}\sigma\, .\]
 We will use that the difference between $Z_1$ and $Z_2$
is small, namely
\begin{equation}\label{DZ}
\begin{split}
\Delta Z(t):=Z_2(t)-Z_1(t)&= \sum_\ell \int_0^t G_{\cdot P_\ell}(t,s) {\rm Re} \langle e^{-{\rm i}s\omega_km^{-1/2}}\gamma_k\nu_k' , \bar V''\partial_{X_\ell} a\rangle {\rm d} s\\
& = \sum_\ell \int_0^t G_{\cdot P_\ell}(t,s) \frac{T^{1/2}}{\omega_k  n^{1/2}}{\rm Re} ( e^{-{\rm i}s\omega_km^{-1/2}}\xi_k \beta_{k,\ell}) 
 {\rm d} s\, ,
 \end{split}
\end{equation}
where $\xi_k=\omega_kT^{-1/2}\gamma_k=\xi_{k,r}+{\rm i}\xi_{k,i}$, with $\xi_{k,r}$ and $\xi_{k,i}$
 independent and standard normal distributed with mean zero and variance one.
The eigenvalue representation \eqref{hb_model} and \eqref{xi_n} show that $\sqrt n\, \omega_k\to \infty$
as $n\to \infty$. Consequently we have $|Z_2-Z_1|\to 0$  as $n\to\infty$, uniformly for all realizations.


{\it Step 3.}
The dependence between $\partial_Pu(X_t,P_t,t)$ and 
${\rm Re}\langle e^{-{\rm i}t\bar V''/m^{1/2}}\varphi_0 , \bar V''\nabla a\rangle$ from the initial data  $\varphi_0$ 
with $Z' = (X',P')$ can by \eqref{expand} and \eqref{gamma_W}   be written 
\begin{equation}\label{first_exp}
\begin{split}
&\mathbb E[\langle {\rm Re}(e^{-{\rm i}t\sqrt{\bar V''}/\sqrt{m}} \varphi_0),
\bar V''\nabla a\cdot \nabla_P u(X_t,P_t,t)\rangle ] \\
&=
\sum_{k,\ell,\ell',\ell''}
\mathbb E\Big[ \langle {\rm Re}(\underbrace{e^{-{\rm i}t\sqrt{\bar V''}/\sqrt{m}}}_{=S_{t0}} \gamma_k\nu_k'),
\bar V''\partial_{X_\ell} a\rangle \times\\
&\qquad \times\int_0^t \underbrace{\Big({\dD}^2_{P_{\ell'}P_\ell} u(X'_t,P'_t,t)G_{P_{\ell'}P_{\ell''}}(t,s)
+ {\dD}^2_{X_{\ell'}P_\ell}u(X'_t,P'_t,t)G_{X_{\ell'}P_{\ell''}}(t,s)\Big)}_{=:U_{\ell,\ell',\ell''}(X_t,P_t, G,t,s)}\times\\
&\qquad\times
\langle {\rm Re}(e^{-{\rm i}s\sqrt{\bar V''}/\sqrt{m}} \gamma_k\nu_k'),
\bar V''\partial_{X_{\ell''}} a\rangle{\rm d} s\Big]\\
\end{split}
\end{equation}
so that  by \eqref{f_def} 
\begin{equation}\label{dupG}
\begin{split}
&\mathbb E[\langle {\rm Re}(e^{-{\rm i}t\sqrt{\bar V''}/\sqrt{m}} \varphi_0),
\bar V''\nabla a\cdot \nabla_P u(X_t,P_t,t)\rangle ] \\
&=
\sum_{k,\ell,\ell',\ell''}
{\rm Re}(\underbrace{e^{{\rm i} t\omega_k m^{-1/2}} \frac{T^{1/2}\beta_{k,\ell}}{\omega_k n^{1/2}}}_{=:\epsilon_{k,\ell}(t)}) \int_0^t  \mathbb E[U_{\ell,\ell',\ell''}(X'_t,P'_t,G,t,s)
|\xi_k^r|^2]
{\rm Re}\big(\epsilon_{k,\ell''}(s)\big){\rm d} s\\
&\quad
+\sum_{k,\ell,\ell',\ell''}
{\rm Im}\big(\epsilon_{k,\ell}(t)\big) \int_0^t  \mathbb E[
U_{\ell,\ell',\ell''}(X'_t,P'_t,G,t,s)|\xi_k^i|^2]
{\rm Im}\big(\epsilon_{k,\ell''}(s)\big){\rm d} s\\
&\quad-\sum_{k,\ell,\ell',\ell''}
{\rm Im}\big(\epsilon_{k,\ell}(t)\big) \int_0^t  \mathbb E[
U_{\ell,\ell',\ell''}(X'_t,P'_t,G,t,s)\xi_k^r\xi_k^i]
{\rm Re}\big(\epsilon_{k,\ell''}(s)\big){\rm d} s\\
&\quad-\sum_{k,\ell,\ell',\ell''}
{\rm Re}\big(\epsilon_{k,\ell}(t)\big)\int_0^t  \mathbb E[
U_{\ell,\ell',\ell''}(X'_t,P'_t,G,t,s)\xi_k^r\xi_k^i]
{\rm Im}\big(\epsilon_{k,\ell''}(s)\big){\rm d} s\, .\\
%
\end{split}
\end{equation}
To determine the expected value $\mathbb E[
U_{\ell,\ell',\ell''}(X'_t,P'_t,G,t,s)\xi_k^r\xi_k^i]$  we use 
 $Z_2=(X_2,P_2)$ in \eqref{perturb} based on $\varphi_2$
 and the Green's function 
 \[
 G_2(t,s) = 
 \begin{bmatrix}
 G_{2,XX}(t,s) & G_{2,XP}(t,s)\\
 G_{2,PX}(t,s) & G_{2,PP}(t,s)
 \end{bmatrix} \in \rset^{2N \times 2N}
 \]
 based on $X_2$ that is defined by
\begin{equation*}
\begin{split}
\partial_t G_{2,X y}(t,s) &=G_{2,P y}(t,s)\, ,\\
\partial_t G_{2,P y}(t,s) 
&= -{\dD}^2\lambda\big(X_2(t)) G_{2,X y} \\
&\quad -\int_s^t\langle \bar V''\cos\big(\frac{(t-r)\bar V''^{1/2}}{m^{1/2}}\big)\nabla a\cdot G_{2,P y}(r,s),\nabla a\rangle {\rm d}r\, ,\ t>s,\ y=X,P,\\
G_2(s,s) &= 
\left[\begin{array}{cc}
{\rm I} & 0\\
0 & {\rm I}\end{array}
\right] \, .
\end{split}
\end{equation*}
The difference of the two Green's functions satisfy by   \eqref{perturb} the perturbation representation
\[
\begin{split}
&G(t,s)-G_2(t,s)\\
&= \int_s^t G(t,r) 
\left[\begin{array}{cc}
0 & 0\\
{\dD}^2\lambda\big(X_2(r)\big)-
\int_0^1{\dD}^2\lambda\big(\tau X_1(r)+(1-\tau)X_2(r)\big){\rm d}\tau & 0
\end{array}
\right] G_2(r,s){\rm d} r\\
&= \int_s^t G(t,r) 
\left[\begin{array}{cc}
0 & 0\\
\int_0^1\int_0^1{\dD}^3\lambda\big(\tau X_2(r)+\tau\sigma \Delta X(r) \big) \tau \Delta X(r) {\rm d}\tau {\rm d}\sigma& 0
\end{array}
\right] G_2(r,s){\rm d} r
\, .
\end{split}
\]
The processes $Z_2$ and $G_2$ do not depend on $\xi_k$.
The expected value can therefore be evaluated using the Jacobian $U'={\dD}U$, with respect to $Z$ and $G$, and the small perturbation $\Delta Z$ in the path from \eqref{DZ} as
\[
\begin{split}
&\mathbb E[
U_{\ell,\ell',\ell''}(X'_t,P'_t,G,t,s)\xi_k^r\xi_k^i]\\
&=\mathbb E[
U_{\ell,\ell',\ell''}(Z_2(t),G_2,t,s)\xi_k^r\xi_k^i]\\
&\quad +\mathbb E[ \int_0^1\int_0^1
U'_{\ell,\ell',\ell''}\big( Z_2(t)+\sigma\tau (Z_1(t)-Z_2(t)),G_2+\tau (G-G_2),t,s\big){\rm d}\tau\times \\
&\qquad\times \xi_k^r\xi_k^i 
\left[\begin{array}{c}
 \sigma \big(Z_1(t)-Z_2(t)\big)\\
 G-G_2\\
 \end{array}
 \right]{\rm d}\sigma]\\
 &=\mathbb E[ U_{\ell,\ell',\ell''}(Z_2(t),t,s)]\underbrace{\mathbb E[\xi_k^r\xi_k^i]}_{=0}
+\mathcal O(\frac{1}{n^{1/2}\eta_n}) 
 \end{split}
\]
with analogous splittings for 
\[\mbox{
$U_{\ell,\ell',\ell''}(X'_t,P'_t,G,t,s)\xi_k^i\xi_k^i$ and 
$U_{\ell,\ell',\ell''}(X'_t,P'_t,G,t,s)\xi_k^r\xi_k^r$,}
\]based on $\mathbb E[\xi_k^r\xi_k^r]=\mathbb E[\xi_k^i\xi_k^i]=1$. Dominated convergence, using the assumption $n^{1/2}\eta_n\to\infty$ in \eqref{xi_n} 
as $n \to \infty$, implies therefore
\[
\begin{split}
&\lim_{\bar n\to\infty}\mathbb E[\langle {\rm Re}(e^{-{\rm i}t\sqrt{\bar V''}/\sqrt{m}} \varphi_0),
\bar V''\nabla a\cdot \nabla_P u(X_t,P_t,t)\rangle ] \\
&=
\lim_{\bar n\to\infty}\sum_{k,\ell,\ell',\ell''}
{\rm Re}\big(\epsilon_{k,\ell}(t)\big) \int_0^t  \mathbb E[
U_{\ell,\ell',\ell''}(X_t,P_t,G,t,s)]
{\rm Re}\big(\epsilon_{k,\ell''}(s)\big){\rm d} s\\
&\quad+\lim_{\bar n\to\infty}\sum_{k,\ell,\ell',\ell''}
{\rm Im}\big(\epsilon_{k,\ell}(t)\big) \int_0^t  \mathbb E[
U_{\ell,\ell',\ell''}(X_t,P_t,G,t,s)]
{\rm Im}\big(\epsilon_{k,\ell''}(s)\big){\rm d} s\\
&=
\lim_{\bar n\to\infty}\sum_{k,\ell,\ell',\ell''}\int_0^t
{\rm Re}\big(\epsilon_{k,\ell}^*(t)\epsilon_{k,\ell''}(s)\big)   \mathbb E[
U_{\ell,\ell',\ell''}(X_t,P_t,G,t,s)]{\rm d} s\, ,\\
\end{split}
\]
which shows that $(X_t,P_t)$ in the limit becomes independent of the small perturbation  caused
by $\langle \gamma_k\nu_k',\bar V''\nabla a\rangle$.
We have 
\[
\begin{split}
\sum_k{\rm Re}\big(\epsilon_{k,\ell}^*(t)\epsilon_{k,\ell''}(s)\big) &=
T{\rm Re} (\sum_k\langle\nu_k', e^{{\rm i}(t-s)\bar V''/\sqrt m}\partial_{X_\ell}a\rangle\langle \nu_k',\bar V''\partial_{X_{\ell''}}a\rangle)\\
&=T\langle \cos\big((t-s)\bar V''/\sqrt m\big)\partial_{X_\ell}a, \bar V''\partial_{X_{\ell''}}a\rangle\, ,
\end{split}
\]
and we obtain then as in the proof of \eqref{fluc_diss}
\[
\begin{split}
& \lim_{\bar n\to\infty}
\sum_{\ell,\ell',\ell''}\mathbb E\Big[\sum_k
{\rm Re}\langle  S_{t0}(\gamma_k^r+{\rm i}\gamma_k^i)\nu_k', \bar V''\partial_{X_\ell} a\rangle
\times\\
&\quad\times\Big({\dD}^2_{P_{\ell'}P_\ell} u \int_0^t G_{P_{\ell'}P_{\ell''}}(t,s) {\rm Re}\langle S_{s0}
(\gamma_k^r+{\rm i}\gamma_k^i)
\nu_k', \bar V''\partial_{\ell''} a \rangle
{\rm d} s\\
&\quad+{\dD}^2_{X_{\ell'}P_\ell} u \int_0^t G_{X_{\ell'}P_{\ell''}}(t,s) {\rm Re}\langle
S_{s0}(\gamma_k^r+{\rm i}\gamma_k^i)\nu_k' , \bar V''\partial_{\ell''} a \rangle
{\rm d} s\Big)\Big]\\
&= T\lim_{\bar n\to\infty} \sum_{\ell,\ell',\ell''}\mathbb E[\int_0^t \big({\dD}^2_{P_{\ell'}P_\ell} u G_{P_{\ell'}P_{\ell''}}(t,s) + {\dD}^2_{X_{\ell'}P_{\ell}} u G_{X_{\ell'}P_{\ell''}}(t,s)
\big)\times\\
&\qquad \times\langle
\cos\big((t-s)\sqrt{\bar V''}/\sqrt{m}\big) \partial_{X_{\ell}} a, \bar V'' \partial_{X_{\ell''}} a\rangle {\rm d}s]\\
&= m^{1/2}{T}\sum_{\ell,\ell',\ell''}
\mathbb E[\int_0^{t/\sqrt{m}} \Big({\dD}^2_{P_{\ell'}P_\ell} u(X_t,P_t,t) G_{P_{\ell'}P_{\ell''}}(t,t-\sqrt{m}\tau) \\
&\qquad + {\dD}^2_{X_{\ell'}P_{\ell}} u(X_t,P_t,t) G_{X_{\ell'}P_{\ell''}}(t,t-\sqrt{m}\tau)\Big)
\langle \cos(\tau\sqrt{\bar V''}) \partial_{X_{\ell}} a, \bar V'' \partial_{X_{\ell''}} a\rangle{\rm d}\tau]\, .\\
\end{split}
\]
Lemma  \ref{fric_lem} with $G(t,t-\sqrt m\tau)=G(t-\sqrt m\tau,t-\sqrt m\tau) + \int_{t-\sqrt m\tau}^t\partial_r G(r,t-\sqrt m\tau) {\rm d} r$ replacing 
\[\big(X(t-\sqrt m \tau),P(t-\sqrt m \tau)\big)=\big((X(t)P(t)\big) - \int_{t-\sqrt m\tau}^t \big(\dot X(s),\dot P(s)\big){\rm d} s\, ,\]
using 
also \eqref{g_ekv22}, $G_{PP}(t-\sqrt m\tau,t-\sqrt m\tau)={\rm I}$ and $G_{XP}(t-\sqrt m\tau,t-\sqrt m\tau)=0$, verifies \eqref{b}:
\[
\begin{split}
& m^{1/2}{T}\sum_{\ell,\ell',\ell''}
\mathbb E[\int_0^{t/\sqrt{m}} \Big({\dD}^2_{P_{\ell'}P_\ell} u(X_t,P_t,t) G_{P_{\ell'}P_{\ell''}}(t,t-\sqrt{m}\tau) \\
&\qquad + {\dD}^2_{X_{\ell'}P_{\ell}} u(X_t,P_t,t) G_{X_{\ell'}P_{\ell''}}(t,t-\sqrt{m}\tau)\Big)
\langle \cos(\tau\sqrt{\bar V''}) \partial_{X_{\ell}} a, \bar V'' \partial_{X_{\ell''}} a\rangle{\rm d}\tau]\\
&= m^{1/2}T\mathbb E[ {\dD}^2_{PP} u(X_t,P_t,t) : \kappa] +\mathcal O(m\log m^{-1})\, .
\end{split}
\]

%
\end{proof}

\section{Analysis of the generalized Langevin equation for $\chi\ll 1$}\label{chi_sec}
In this section we study a time scale separation for system particles and fast heat bath particles 
due to a stiff heat bath obtained by changing the scaling to
\begin{equation}\label{chi_cond}
\begin{split}
\bar V'' &= \chi^{-1}\tilde V''\, ,\\
a(X) &=\chi^\delta \tilde a(X)\, ,\\
\varphi_0&=\sum_k \gamma_k\nu_k'=\chi^{1/2}\sum_k\tilde \gamma_k\nu_k'=\chi^{1/2}\tilde\varphi_0\, ,\\
\tilde \gamma_k &\sim N(0,T/\tilde\omega_k^2)\, ,\\
\tilde\omega_k^2 &\mbox{ are the eigenvalues of $\tilde V''$}\, ,\\
\end{split}
\end{equation}
and assume that 
\begin{equation}\label{chi_cond2}
\mbox{$m$ and the elements of $\tilde V''$ and $ \tilde a$ are of size one while $\chi\ll 1$.}
\end{equation}
The nonlinear generalized Langevin equation then becomes
\begin{equation*}\label{4.5}
\begin{split}
\ddot X_t&= -\nabla \lambda(X_t) - \chi^{2\delta-1}
\int_0^t \langle\tilde V'' \cos\big(\frac{(t-s) \tilde V''^{1/2}}{(\chi m)^{1/2}}\big) \dot{ \tilde a}(X_s), \nabla \tilde 
a(X_t)\rangle{\rm d} s\\
&\quad + \chi^{\delta -1/2}
{\rm Re}\big\langle \tilde V'' e^{-{\rm i} \tilde V''^{1/2} (\chi m)^{-1/2}t}\tilde \varphi_0, \nabla \tilde a(X_t)\big\rangle\, .
\end{split}
\end{equation*}
The analysis in Section \ref{sec3} can be applied by replacing $m$ by $\chi m$; $\bar V''$ by $\tilde V''$; and 
$a$ by $\tilde a$ and we obtain
\begin{theorem}\label{thm2} 
Let $m_0:=\chi^{2\delta-1/2} m^{1/2}$ with $\delta>1/4$ and assume that \eqref{chi_cond}-\eqref{chi_cond2} hold together with
the assumptions in Theorem \ref{thm1}, then
the Langevin dynamics 
\begin{equation}\label{Ito_langevin3}
\begin{split}
{\rm d}X_L(t) &= P_L(t){\rm d}t\\
{\rm d}P_L(t) &= -\nabla \lambda\big(X_L(t)\big){\rm d}t - m_0\kappa P_L(t){\rm d}t  + (2 m_0\kappa T)^{1/2}\ {\rm d} W(t)\, ,\\
X_L(0) &=X(0)\, ,\\
P_L(0) &=P(0)\, \\
\end{split}
\end{equation}
 approximates the Hamiltonian dynamics \eqref{1.2} and \eqref{1.5},
with the error estimate 
\[\begin{split}
&\Big|\mathbb E[ g(X_{t_*},P_{t_*})\ |\ X_0,P_0]
-\mathbb E[ g\big(X_L(t_*,P_L(t_*)\big)\ |\ \big(X_L(0),P_L(0)\big)=(X_0,P_0)] \Big|\\
&=
\mathcal O(m_0\chi^{1/2}\log\chi^{-1})
=\mathcal O(\chi^{2\delta}\log\chi^{-1})\, , 
\end{split}\]
where the friction matrix is determined by the force 
$\tilde F_{\ell j}=\lim_{n\to\infty} (\tilde V''\partial_{X_\ell}\tilde a)(j)$ as
\begin{equation}\label{kap2}
\kappa_{\ell \ell'}= \frac{1}{4\pi c^3}\big(\sum_{j\in E_\infty} \tilde F_{\ell j}\big)
\big(\sum_{j\in E_\infty} \tilde F_{\ell' j}\big)\, .
\end{equation}
\end{theorem}


\section{Molecular dynamics approximation of a quantum system}\label{sec2}
The purpose of this section is to present a molecular dynamics approximation for
 observables of a quantum particle system consisting of nuclei and electrons coupled to a heat bath.
  The observables may include correlations in time.
 The first subsection provides background to quantum observables approximated by molecular dynamics in the canonical ensemble. The next subsection combines these quantum approximation results of with the 
 classical Langevin approximation in Theorems \ref{thm1} and \ref{thm2}.
 
 \subsection{Canonical quantum observables approximated by molecular dynamics}

 The quantum formulation is based on  wave functions $\Phi:\rset^N\times\rset^n\to\mathbb C^d$
 and the Hamiltonian
 \[
-\frac{M_s^{-1}}{2}{\rm I}\Delta_X -\frac{M_b^{-1}}{2}{\rm I}\Delta_x + V(X)+ V_b(x,X)\, ,
 \]
 where $X\in\rset^N$ and $x\in\rset^n$ are the nuclei coordinates of the system and heat bath positions, respectively,
 and $M_s$  and $M_b$ are the diagonal matrices  of the mass of the system nuclei and heat bath nuclei, respectively,
 measured in units of the electron mass.
 The functions $V:\rset^N\to \mathbb C^{d^2}$ and $V_b:\rset^n\times\rset^N\to  \mathbb C^{d^2}$
 are finite difference approximations of the electron kinetic energy and nuclei--nuclei, nuclei--electron and electron--electron interactions, related to the system and the heat bath. The matrix ${\rm I}$ is the identity on
 $\mathbb C^d$. This simplification to replace the Laplacians for the electron kinetic energy
 by difference approximations makes it easier to derive the classical limit. Another simplification is to change coordinates $M^{1/2}\bar X=M_s^{1/2}X$ and $M^{1/2}\bar x=M_b^{1/2}x$ and let $\tilde x=(\bar X,\bar x)\in\rset^{N+n}$,
 where $M\gg 1$ is a reference nuclei--electron mass ratio.
 In these coordinates the Hamiltonian takes the form
 \[
 \hat{H}= -\frac{1}{2M}{\rm I} \Delta_{\tilde x}+ v({\tilde x})\, ,
 \]
 where $v({\tilde x}):= V(M_s^{-1/2}M^{1/2}\bar X) + V_b(M_b^{-1/2}M^{1/2}\bar x,M_s^{-1/2}M^{1/2}\bar X)$.
 We will use the eigenvalues $\lambda_k({\tilde x})\in\rset$ and
 eigenvectors $\psi_k({\tilde x})$ of the Hermitian matrix $v({\tilde x})$ defined by
 \[
 v({\tilde x})\psi_k({\tilde x}) = \lambda_k({\tilde x})\psi_k({\tilde x})\, .
 \]
%
 We assume that the eigenvalues 
satisfy
\begin{equation}\label{lam_cond}
\begin{split}
&\lambda_1({\tilde x})<\lambda_2({\tilde x})<\ldots<\lambda_d({\tilde x})\, ,\\
&\lambda_1({\tilde x})\rightarrow\infty \mbox{ as } |{\tilde x}|\rightarrow \infty\, .
\end{split}
\end{equation}
The first assumption is in order to have differentiable eigenvectors and the second condition implies
 that the spectrum of $\hat{ H}$ is discrete, see \cite{dell}.

The aim here is to study canonical quantum observables, including correlations in time, namely
\[
{\rm trace}( \hat { A}_\tau\hat C_0)= 
\sum_{n=1}^\infty ( \Phi_n, \hat { A}_\tau\hat C_0 \Phi_n)\, ,
\]
where $\{\Phi_n\}_{n=1}^\infty$ is a normalized basis of $L^2({\rm d}{\tilde x})$, e.g. the set if normalized eigenfunctions to
$\hat{H}$ and $( f,g):=\int_{\rset^{N+n}} f^*({\tilde x})g({\tilde x}){\rm d}{\tilde x}$. An operator
$\hat B$ is the Weyl quantization  that maps $L^2(\rset^{N+n})$ to  $[L^2(\rset^{N+n})]^d$
and is defined, from a $d\times d$ matrix valued symbol $B:\rset^{N+n}\to \mathbb C^{d^2}$ in the Schwartz class, by
\[
\hat{B}\Phi({\tilde x})= (\frac{M^{1/2}}{2\pi})^{N+n}\int_{\rset^{2(N+n)}} e^{{\rm i}M^{1/2} ({\tilde x}-y)\cdot {\tilde p}}B(\frac{{\tilde x}+y}{2},{\tilde p}) \Phi(y){\rm d}{\tilde p}{\rm d} y\, .
\]
For instance, we have $\widehat{\frac{|{\tilde p}|^2}{2} +v({\tilde x})}=\hat{H}$.
 The time dependent operator $\hat B_\tau$  is defined by
\begin{equation}\label{a_t_def}
\hat B_\tau:=e^{{\rm i}\tau M^{1/2}\hat{ H}}\hat{ B}e^{-{\rm i}\tau M^{1/2}\hat{ H}}\, , \tau\in\rset,
\end{equation}
which implies the von Neumann-Heisenberg equation
\begin{equation}\label{vNH}
\frac{{\rm d}}{{\rm d}\tau} \hat B_\tau = {\rm i} M^{1/2}[\hat H,\hat B_\tau]
\end{equation}
where $[\hat H,\hat B_\tau]:=\hat H\hat B_\tau-\hat B_\tau\hat H$ is the commutator. 
The example of the observable for the diffusion constant \[
\frac{1}{6\tau}\frac{3}{N+n}\sum_{k=1}^{(N+n)/3} |{\tilde x}_k(\tau)-{\tilde x}_k(0)|^2=\frac{1}{2(N+n)\tau}\big(
|{\tilde x}({\tau})|^2 +|{\tilde x}(0)|^2-2{\tilde x}(\tau)\cdot {\tilde x}(0)\big)\]
uses the time-correlation $\hat {\tilde x}(\tau)\cdot \hat {\tilde x}(0)$ where  $\hat A_\tau=\hat {\tilde x}_\tau
\rm I$ and $\hat C_0=\hat {\tilde x}_0\rm I$ and
\[
\hat {\tilde x}_\tau\cdot \hat {\tilde x}_0=\sum_{k=1}^{(N+n)/3}\sum_{j=1}^3e^{{\rm i}\tau M^{1/2}\hat{ H}}\hat {\tilde x}_{k_j} e^{-{\rm i}
\tau M^{1/2}\hat{ H}}  \hat {\tilde x}_{k_j}\, .\]

 A main tool to determine the classical limit is to diagonalize \eqref{vNH},
 which is based on the following composition of Weyl quantizations:
the symbol $C$ for the product of two Weyl operators $\hat A\hat B=\hat C$
is determined by
\begin{equation}\label{eq:composition}
C({\tilde x},{\tilde p}) = e^{\frac{i}{2M^{1/2}} (\nabla_{ {\tilde x}'}\cdot\nabla_{{\tilde p}}-\nabla_{ {\tilde x}}\cdot \nabla_{ {\tilde p}'})}A( {\tilde x}, {\tilde p})
B( {\tilde x}', {\tilde p}')\Big|_{
{\tiny \begin{array}{c}
 {\tilde x}= {\tilde x}'\\
 {\tilde p}= {\tilde p}'
\end{array}
}} =:(A\# B)({\tilde x},{\tilde p})\, ,
\end{equation}
see \cite{zworski}. 
 Assume that $\Psi:\rset^{N+n}\rightarrow \mathbb C^{d^2}$ and $\Psi({\tilde x})$ is any unitary matrix  with the Hermitian transpose
 $\Psi^*({\tilde x})$ and define
$ {\bar A}:\rset^{N+n}\times [0,\infty)\rightarrow \mathbb C^{d^2}$ by
\[
\hat A_{\tau} = \hat \Psi({\tilde x}) \hat{\bar A}_{\tau}\hat\Psi^*({\tilde x})
\]
so that
\[
\hat{\bar A}_{\tau} = \hat \Psi^*({\tilde x}) \hat{ A}_{\tau}\hat\Psi({\tilde x})\, .
\]
Then 
\[
[\hat H,\hat A_{\tau}]= \hat\Psi[\hat\Psi^*\hat H\hat\Psi, \hat{\bar A}_{\tau}]\hat\Psi^*
\]
and consequently
\[
\partial_{\tau}\hat{\bar A}_{\tau}  = {\rm i}M^{1/2}[\hat\Psi^*\hat H\hat\Psi, \hat{\bar A}_{\tau}]\, .
\]
The composition rule \eqref{eq:composition} implies $\hat\Psi^*\hat H\hat\Psi=(\Psi^*\#H\#\Psi)^{\widehat{}}$
and $\bar A_\tau=\Psi^*\#A_\tau\#\Psi$.
The next step is to determine $\Psi$ so that 
\[
\bar H:=\Psi^*\#H\#\Psi
\]
is  almost diagonal. Having $\bar H$ diagonal implies that $\hat{\bar H}$ is diagonal
and then $\hat{\bar A}$ remains diagonal if it initially were diagonal, since then
\[
\begin{split}
\frac{{\rm d}}{{\rm d}{\tau}} \hat{\bar A}_{jk}({\tau})
& = iM^{1/2}\big(\hat{\bar H}_{jj}\hat{\bar A}_{jk}({\tau}) -
\hat{\bar A}_{jk}({\tau})\hat{\bar H}_{kk}\big)=0\, , \quad \mbox{ for } j\ne k.\\
\end{split}
\]

The composition rule \eqref{eq:composition} with 
\[
H({\tilde x},{\tilde p})=\frac{|{\tilde p}|^2}{2} {\rm I} + v({\tilde x})
\]
implies that
\[
\begin{split}
\bar H&= \Psi^*\# H\# \Psi \\
&= \frac{|{\tilde p}|^2}{2} {\rm I} + \Psi^*v\Psi + \frac{1}{4M} \nabla\Psi^*\cdot \nabla\Psi\\
&=\Psi^*(\frac{|{\tilde p}|^2}{2} {\rm I} + v + \frac{1}{4M} \Psi\nabla\Psi^*\cdot \nabla\Psi\Psi^*)\Psi\, ,
\end{split}
\]
as verified in~\cite[Lemma 3.1]{KSS}.
Therefore, the aim is to choose the unitary matrix $\Psi$  so that it becomes an approximate solution to the nonlinear eigenvalue problem
\begin{equation}\label{eigen_nonlin}
\big(v + \frac{1}{4M} \Psi\nabla\Psi^*\cdot \nabla\Psi\Psi^*\big)\Psi= \Psi\bar\Lambda
\end{equation}
in the sense that
\begin{equation}\label{eigen_nonlin2}
\big(v + \frac{1}{4M} \Psi\nabla\Psi^*\cdot \nabla\Psi\Psi^*\big)\Psi= \Psi\bar\Lambda +\mathcal O(M^{-2})
\end{equation}
where $\bar\Lambda:\rset^{N+n}\to \mathbb C^{d\times d}$ is diagonal, that is
\[
\bar\Lambda_{jk}({\tilde x})=\left\{\begin{array}{cc}
0 & j\ne k\\
\bar\lambda_j({\tilde x}) & j=k
\end{array}\right. \, . 
\]
Such a solution $\Psi$ then implies
\begin{equation}\label{bar_H}
\bar H({\tilde x},{\tilde p})= \frac{|{\tilde p}|^2}{2} {\rm I} +\bar\Lambda({\tilde x}) + r_0({\tilde x})\, ,
\end{equation}
where the remainder satisfies $\|r_0\|_{L^{\infty}(R^{N+n})}=\mathcal O(M^{-2})$.
A solution, $\Psi$, to this nonlinear eigenvalue problem  is an $\mathcal O(M^{-1})$ perturbation of the eigenvectors to $v({\tilde x})$
provided the eigenvalues do not cross and $M$ is sufficiently large. 
The work \cite[(3.18)]{KSS} shows that \eqref{eigen_nonlin2} has a solution $\Psi$, 
if $v$ is twice differentiable, the eigenvalues of $v$ are distinct 
and $M$ is sufficiently large.

The canonical ensemble is typically based on ${\rm trace}(\hat Ce^{-\hat{\bar H}/T})/ 
{\rm trace}(e^{-\hat{\bar H}/T})$. We will instead use the related
${\rm trace}(\hat C \widehat{e^{-\bar H/T}})/ 
{\rm trace}(\widehat{e^{- \bar H/T}})$.
%
If the density operators $e^{-\hat{\bar H}/T}$ and $\widehat{e^{-\bar H/T}}$ would differ only little
it would not matter which one we use as a reference. 
Since we do not know if this difference is small in the case of a large
number of particles, we may ask which density operator to use.  The density operator $\hat\rho_q=e^{-\hat{\bar H}/T}$
is a time-independent solution to the quantum Liouville-von Neumann equation
\[
\partial_t \hat\rho_t= iM^{1/2}[\hat\rho_t, \hat{\bar H}]
\]
while the classical Gibbs density $e^{-\bar H/T}$ is a time-independent solution to the classical Liouville  equation
\[
\partial_t \bar\rho_t= -\{\bar\rho_t, \bar H\}\, ,
\]
with the Poisson bracket  in the right hand side.
The corresponding density matrix symbol $\rho_q$ is not a time-independent solution to the classical Liouville equation, since \[0={\rm i}
M^{1/2}(\rho_q\#\bar H-\bar H\#\rho_q)\ne \{\rho_q,\bar H\}\, ,\]
and the classical Gibbs density
is not a time-independent solution to the quantum Liouville-von Neumann equation, since
${\rm i}M^{1/2}(e^{-\bar H/T}\#\bar H-\bar H\#e^{-\bar H/T})\ne \{ e^{-\bar H/T},\bar H\}=0$.
However, it is shown in \cite{KSS} that a solution to
 the quantum Liouville equation $\hat\rho_t$
 with initial data $\rho_0=e^{-\bar H/T}$ 
 generates only a small time dependent perturbation on observables up to time $t<M$,
 which motivates our use of $\widehat{e^{-\bar H/T}}$.

The following result for approximating non equilibrium quantum observables
by classical molecular dynamics observables is proved in \cite{KSS}.
\begin{theorem}\label{gibbs_corr_thm_analytic}
 Assume that $v$ satisfies \eqref{lam_cond}, the $d\times d$ matrices $\bA\equiv\bA_0$ and $\bB$ are diagonal, the $d\times d$ matrix valued Hamiltonian $ H$ has distinct eigenvalues, %
and that there is a constant $C$ such that
 \begin{equation*}\label{R77}
\begin{split}
\sum_{|\alpha|\le 2}\|\partial^\alpha_{\tilde x} \psi_k\|_{L^\infty(\rset^{N})} &\le C\, ,\quad k=1,\ldots,d\, ,\\
\max_i\sum_{|\alpha|\le 3}\|\partial^\alpha_{\tilde x} \partial_{{\tilde x}_i}\bar\lambda_j\|_{L^\infty(\rset^{N+n})} &\le C\, ,\\
\sum_{|\alpha|\le 3}
\|\partial_{\tilde z}^\alpha\bA_{jj}\|_{L^2(\rset^{2(N+n)})}  &\le C\, ,\\
\|\bB({\tilde z})e^{-\bar H({\tilde z})/T}\|_{L^2(\rset^{2(N+n)})} &\le C\, ,\\
\end{split}
\end{equation*} 
then there is a constant $C'$, depending on $C$, such that the canonical ensemble average satisfies
\begin{equation*}\label{G_corr_unif1}
\begin{split}
|\frac{{\rm trace}\big(\hat{ A}_\tau \hat\Psi{ ({\bB}e^{-{\bar H}/T})^{\widehat{}}} \, \hat\Psi^*\big)}{
{\rm trace}(\hat\Psi\widehat{e^{-{\bar H}/T}}\hat\Psi^*)}
-
  \sum_{j=1}^d \int_{\rset^{2(N+n)}}
\frac{ \bA_{jj}({\tilde z}^j_\tau ({\tilde z}_0))\bB_{jj}({\tilde z}_0)  e^{-\bar{H}_{jj}({\tilde z}_0)/T}}{\sum_{k=1}^d\int_{\rset^{2(N+n)}}e^{-\bar{H}_{kk}({\tilde z})/T} \Rd {\tilde z}} \Rd {\tilde z}_0|
\le \frac{C'}{M}\, ,
\end{split}
\end{equation*}
where ${\tilde z}^j_{\tau}=({\tilde x}_{\tau},{\tilde p}_{\tau})$ is the solution to the Hamiltonian system
\begin{equation}\label{HS}
\begin{split}
\dot {\tilde x}_{\tau} &= {\tilde p}_{\tau}\\
\dot {\tilde p}_{\tau} &= -\nabla\bar \lambda_j({\tilde x}_{\tau}), \quad {\tau}>0\, ,
\end{split}
\end{equation}
based on the Hamiltonian $\bar H_{jj}({\tilde x},{\tilde p})=|{\tilde p}|^2/2 + \bar\lambda_j({\tilde x})$,
with initial data $({\tilde x}_0,{\tilde p}_0)=\tilde z_0\in\rset^{2(N+n)}$.
\end{theorem}

\subsection{Langevin dynamics derived from quantum mechanics}
Assume that the potential $V(X)$ has the eigenvalues $\lambda_j(X)$ and eigenvectors $\psi_j(X), \ j=1,\ldots, d$. The eigenvalues of the
potential $v(\tilde x(X,x))=V(X)+V_b(x,X)$ 
will to leading order in $\|V_b(\cdot,X)\|$ be given by
$\lambda_j(X)+ \psi_j^*V_b(x,X)\psi_j$.
We assume now that all coupling potentials $\psi_j^*V_b(x,X)\psi_j$
satisfy the weak coupling assumptions
\eqref{2} and \eqref{3}, namely \[\min_x \psi_j^*(X)V_b(x,X)\psi_j(X)=\psi_j^*(X)V_b(a_j(X),X)\psi_j(X)=0\]
so that
\[
\bar\lambda_j\big(\tilde x(x,X)\big)=\lambda_j(X)+\langle x-a_j(X),\bar V_j''\big(x-a_j(X)\big)\rangle
\]
where $\bar V''_j=\psi_j^* {V_b}_{xx}''\psi_j$ is a constant matrix as in \eqref{3} and each coupling $\psi_j^*V_b\psi_j$ for $j=1,\ldots, d$ yields one equilibrium position $a_j(X)$ and one coupling matrix $\bar V_j''$. 
The eigenvalues of the Hamiltonian symbol
\[
(\frac{M}{2} P\cdot M_s^{-1}P + \frac{M}{2} p\cdot M_b^{-1}p)\, {\rm I} + V(X)+V_b(x,X)
\]
are then
\[
\begin{split}
\bar H_{jj}(X,P,x,p) &=\frac{M}{2} P\cdot M_s^{-1}P + \frac{M}{2} p\cdot M_b^{-1}p
+\bar\lambda_j\big(\tilde x(x,X)\big)\\
&=\frac{M}{2} P\cdot M_s^{-1}P + \frac{M}{2} p\cdot M_b^{-1}p
+\lambda_j(X)+\langle x-a_j(X),\bar V_j''\big(x-a_j(X)\big)\rangle\, .
\end{split}
\]
Let $m$ 
be the mass ratio for a reference heat bath nuclei to a reference system nuclei.
Theorems \ref{thm1} and \ref{thm2} show the classical dynamics provided by the
 Hamiltonians $\bar H_{jj}$ are accurately approximated by the Langevin dynamics
\begin{equation}\label{qm_langevin}
\begin{split}
    {\rm d} X_t&= M M_s^{-1} P_t{\rm d}t\, ,\\
     {\rm d} P_t&= -\nabla\lambda_j(X_t) {\rm d}t - m^{1/2}\kappa^j M_s^{-1}P_t{\rm d}t + \sqrt{2m^{1/2}\kappa^j T}{\rm d}W_t\, ,\\
\end{split}
\end{equation}
where $\kappa^j_{\ell\ell'} =\frac{1}{4\pi c_j^3} \langle \bar V_j''\partial_{X_\ell}a_j,\bar V_j''\partial_{X_\ell'}a_j\rangle$, for $m\ll 1$. The next step is to relate this approximation by Langevin dynamics also to quantum observables in the canonical ensemble. In particular we need to obtain the Gibbs distribution of the heat bath particles from the quantum observables in the canonical ensemble.

Define the probability, $q_j$, to be in electron state $j$  as
\begin{equation}\label{q_def}
q_j:=\frac{\int_{\mathbb R^{2(N+n)}}e^{-\bar H_{jj}({\tilde z})/T} \Rd {\tilde z}}{
\sum_{k=1}^d\int_{\mathbb R^{2(N+n)}}e^{-\bar H_{kk}({\tilde z}')/T} \Rd {\tilde z}'}\, ,\quad j=1,\ldots, d\, ,
\end{equation}
then  the molecular dynamics observable becomes a sum of observables in the different electron states
with initial Gibbs distribution, namely
\[
\begin{split}
  &\sum_{j=1}^d \int_{\rset^{2(N+n)}}
\frac{ \bA_{jj}({\tilde z}^j_\tau ({\tilde z}_0))\bB_{jj}({\tilde z}_0)  e^{-\bar{H}_{jj}({\tilde z}_0)/T}}{\sum_{k=1}^d\int_{\rset^{2(N+n)}}e^{-\bar{H}_{kk}({\tilde z})/T} \Rd {\tilde z}} \Rd {\tilde z}_0\\
&=
\sum_{j=1}^d
q_j \int_{\rset^{2(N+n)}}
\frac{ \bA_{jj}({\tilde z}^j_\tau ({\tilde z}_0))\bB_{jj}({\tilde z}_0)  e^{-\bar{H}_{jj}({\tilde z}_0)/T}}{\int_{\rset^{2(N+n)}}e^{-\bar{H}_{jj}({\tilde z})/T} \Rd {\tilde z}} \Rd {\tilde z}_0\, .
\end{split}
\]
For instance if only the ground state matters we have $q_1=1$ and $q_k=0, \ k=2,3,4,\ldots, d$.

We simplify by letting the system nuclei have the same mass $M$ and the heat bath nuclei the same mass $mM$. Then the diagonalized Hamiltonian can by  \eqref{star} be written as
\[
\bar H_{jj}({\tilde z})=\frac{|P|^2}{2} + \lambda_j(X) + \frac{|p|^2}{2m} +\langle(x-a_j(X),\bar V_j''(x-a_j(X)\rangle\, ,
\]
where ${\tilde z}=(X,P,x,p)\in\rset^N\times\rset^N\times\rset^n\times\rset^n$.
Assume  that the observables $\bar A_{jj}$ and $\bar B_{jj}$  
only depend on the system coordinates $X$ and $P$,
then the classical molecular dynamics approximation of the canonical quantum observables in Theorem \ref{gibbs_corr_thm_analytic} satisfies 
\[
\begin{split}
&\int_{\rset^{2(N+n)}} q_j\bA_{jj}\big(X^j_\tau ({\tilde z}_0),P^j_\tau ({\tilde z}_0)\big)\bB_{jj}(X_0,P_0) 
\frac{e^{-\bar{H}_{jj}({\tilde z}_0)/T}}{\int_{\rset^{2(N+n)}}e^{-\bar{H}_{jj}({\tilde z})/T} \Rd {\tilde z}} \Rd {\tilde z}_0\\
&=\int_{\rset^{2n}}
\int_{\rset^{2N}} q_j\bA_{jj}\big(X^j_\tau ({\tilde z}_0),P^j_\tau ({\tilde z}_0)\big)\bB_{jj}(X_0,P_0) \times\\
&\qquad\times\frac{e^{-(\frac{|P_0|^2}{2} + \lambda_j(X_0))/T } }{  \int_{\rset^{2N}}
e^{-(\frac{|P|^2}{2} + \lambda_j(X))/T} \Rd X\Rd P}
\frac{e^{-(\frac{|p|^2}{2m} + \langle x-a_j(X_0),\bar V_j''(x-a_j(X_0))/T}}{\int_{\rset^{2n}}
e^{-(\frac{|p|^2}{2m} + \langle x-a_j(X_0),\bar V_j''(x-a_k(X_0))/T} \Rd x\Rd p} \Rd X_0\Rd P_0\Rd x\Rd p\\
&=\mathbb E[
\int_{\rset^{2N}}q_j \bA_{jj}\big(X^j_\tau ({\tilde z}_0),P^j_\tau ({\tilde z}_0)\big)\bB_{jj}(X_0,P_0) 
\frac{e^{-(\frac{|P_0|^2}{2} + \lambda_j(X_0))/T } }{\int_{\rset^{2N}}
e^{-(\frac{|P|^2}{2} + \lambda_j(X))/T} \Rd X\Rd P}
\Rd X_0\Rd P_0]\, ,\\
\end{split}
\]
where the expected value is with respect to the Gibbs measure 
\begin{equation}\label{gibbs_sampling}
\frac{e^{-(\frac{|p|^2}{2m} + \langle x-a_j(X_0),\bar V_j''(x-a_j(X_0))/T}}{\int_{\rset^{2n}}
e^{-(\frac{|p|^2}{2m} + \langle x-a_j(X_0),\bar V_j''(x-a_k(X_0))/T} \Rd x\Rd p}
\end{equation}
of the heat bath coordinates
conditioned on the initial system coordinates.
We note that the  Gibbs measure \eqref{gibbs_sampling}
is the invariant measure used to sample the initial heat bath configurations in theorems \ref{thm1} and \ref{thm2}. Therefore the combintion of 
theorems \ref{thm1}, \ref{thm2}  and \ref{gibbs_corr_thm_analytic} show that 
canonical observables for a quantum system coupled to a heat bath
can be accurately approximated by
 Langevin dynamics, with the friction coefficient determined by \eqref{kap1} and \eqref{kap2}, including several electron surfaces $\lambda_j, \ j=1,\ldots, d$.
\begin{theorem}\label{kombin}
Suppose that  
the assumptions in Theorems \ref{gibbs_corr_thm_analytic} and  (\ref{thm1} or \ref{thm2})  hold and the observables $\bar A_{jj}(\cdot)$ and $\bar B_{jj}(\cdot)$ depend only on the system coordinates $(X_0,P_0)$,
then
\[
\begin{split}
&\frac{{\rm trace}\big(\hat{ A}_\tau \hat\Psi{ \widehat{({\bB}e^{-{\bar H}/T}})} \hat\Psi^*\big)}{
{\rm trace}(\hat\Psi\widehat{e^{-{\bar H}/T}}\hat\Psi^*)}
=\frac{
{\rm trace}
\Big(\big(\bar A _\tau(X,P)\big)^{\widehat{}}\, 
{ \big({\bar B (X,P)}e^{-{\bar H}/T}}\big)^{\widehat{}}\, 
\Big)
}{
{\rm trace}(\widehat{e^{-{\bar H}/T}})}\\
&=  \sum_{j=1}^d 
\mathbb E[ q_j
\int_{\rset^{2N}} \bA_{jj}\big( X^j_\tau (X_0,P_0), P^j_\tau (X_0,P_0)\big)\bB_{jj}(X_0,P_0) 
\times\\&\qquad\times
\frac{e^{-(\frac{|P_0|^2}{2} + \lambda_j(X_0))/T } }{\int_{\rset^{2N}}
e^{-(\frac{|P|^2}{2} + \lambda_j(X))/T} \Rd X\Rd P}
\Rd X_0\Rd P_0]\\
%
&\quad +\mathcal O(\bar M^{-1}+\bar\chi)\, ,
\end{split}
\]
where 
\[
\begin{split}
&\mbox{$(\bar\chi,\bar m,\bar M):=(m\log m^{-1},m^{1/2}, mM)$  for $m\to 0^+$ in Theorem \ref{thm1} or}\\ 
&\mbox{$(\bar\chi,\bar m,\bar M):= (\chi^{2\delta}\log\chi^{-1},\chi^{2\delta-1/2}m^{1/2},M)\big)$  for $\chi\to 0^+$ in Theorem \ref{thm2},}
\end{split}
\]
and $( X^j_t, P^j_t)$, for $t>0$, is the solution to the Langevin equation
\begin{equation}\label{thm_lang}
\begin{split}
{\rm d} X^j_t &=  P^j_t{\rm d}t\\
{\rm d} P_t^j &= -\nabla \lambda_j\big( X^j_t\big){\rm d}t - \bar m\kappa^j  P^j_t{\rm d}t  + (2 \bar m\kappa^j T)^{1/2}\ {\rm d} W_t\, ,\\
\end{split}
\end{equation}
with initial data $( X^j_0, P^j_0)=(X_0,P_0)$, 
\[
\kappa^j_{\ell\ell'} =\frac{1}{4\pi c_j^3} \langle \bar V_j''\partial_{X_\ell}a_j,\bar V_j''\partial_{X_{\ell'}}a_j\rangle
\]
and $c_j$ set by the density of heat bath states at zero frequency in \eqref{dof}.
The probability $q_j$ to be in electron state $j$ is determined by \eqref{q_def} and
the expected value is with respect to the  Wiener process $W$, with $N$ independent components.
%
\end{theorem}
\section{Numerical example}\label{sec_num}
We consider a single heavy particle in $\mathbb{R}^3$,
the nearest neighbour lattice interaction $\bar V''$, cf.~\eqref{lattice_def}, $\lambda(X) = |X|^2/2$, $c=1$
and 
\[
\beta_\ell(\boldsymbol \omega) = \boldsymbol{1}_{|\boldsymbol \omega | \le 1} \prod_{i=1}^3 \pi^{1/2}(4-\omega_i^2)^{1/4}\quad \text{for all} \quad \ell \in \{1,2,3\}\, , 
\]
implying by~\eqref{eq:fOmegaDef} that
\[
f(\boldsymbol \omega, \ell, \ell') = \boldsymbol{1}_{|\boldsymbol \omega | \le 1}  \quad \text{for all } \quad \ell,\ell' \in \{1,2,3\}\, .
\]
Since this contradicts that $\beta$ is twice differentiable, an assumption that was used in the proof of 
Lemma~\ref{fric_lem}, let us demonstrate that said lemma also applies in the current setting 
with jump-discontinuous $\beta$:
\begin{equation}\label{f_kappa2}
\begin{split}
\lim_{m\to 0^+}&\lim_{n\to\infty}
m^{-1/2}\int_0^t \langle \bar V'' \cos\big(\frac{(t-s)\sqrt{\bar V''}}{\sqrt{m}}\big) \dot a(X_s), \partial_{X^\ell} a\rangle {\rm d} s\\
&= \sum_{\ell'} \lim_{m\to 0^+} m^{-1/2}\int_0^t  \int_{\mathbb R^3} \cos\left(\frac{(t-s)\omega}{\sqrt{m}}\right) f(\boldsymbol{\omega},\ell,\ell') 
\frac{{\rm d}\boldsymbol{\omega}}{\omega^2}{\rm d}\tau \dot X_s^{\ell'} {\rm d} s\\
&= \sum_{\ell'} \lim_{m\to 0^+} m^{-1/2} \int_0^t  \int_{0}^1 \int_0^{2\pi} \int_0^\pi \cos\left(\frac{(t-s)\omega}{\sqrt{m}}\right) \sin(\theta) {\rm d}\theta {\rm d}\alpha {\rm d}\omega  \dot X_s^{\ell'} {\rm d} s\\
&=  4 \pi \sum_{\ell'} \lim_{m\to 0^+} \int_0^t  \frac{\sin\left( (t-s)/\sqrt{m} \right)}{(t-s)} \, \dot X_s^{\ell'} {\rm d} s\\
&=  4 \pi \sum_{\ell'} \lim_{m\to 0^+} \int_0^{tm^{-1/2}}  \frac{\sin\left( \tau \right)}{\tau} \, \dot X_{t-m^{1/2}\tau}^{\ell'} {\rm d} \tau \\
&= 2 \pi^2 \sum_{\ell'} \dot X_t^{\ell'}\, ,
\end{split}
\end{equation}
where the last equality 
follows from $\int_0^\infty \frac{\sin\tau}{\tau}{\rm d}\tau=\pi/2$ and $\dot X_t$ being continuous with respect to $t$.
We conclude that 
\begin{equation}\label{kappaNumEx}
\kappa = 2\pi^2 \begin{bmatrix}1 \\ 1 \\1 \end{bmatrix} \begin{bmatrix}1 & 1 &1 \end{bmatrix} \, .
\end{equation}

To prove the convergence~\eqref{convRateH}, observe first that
\[
K_{\infty}^{\ell \ell'}(\tau)= \int_{\mathbb{R}^3} \cos(\tau \omega) f(\boldsymbol \omega) \frac{{\rm d} \omega}{\omega^2}
= 4\pi \frac{\sin(\tau)}{\tau}. 
\]
Introducing the mesh $\tau_k = 2k \pi$ and  
\[
\dot Y^\ell_{t-\tau \sqrt{m}} := \sum_{k=0}^{\lfloor t/(\sqrt{m}2\pi) \rfloor} \boldsymbol{1}_{\tau \in [\tau_k, \tau_{k+1})} \dot X^\ell_{t- \tau_k\sqrt{m}}\, ,
\]
it follows that for any $\tau \in [\tau_k, \tau_{k+1})$,
there exists a random $s \in [\tau_k, \tau_{k+1})$ such that  
\[
|\dot X_{t-\tau \sqrt{m}} - \dot Y_{t-\tau \sqrt{m}}| = |\ddot{X}_{t-s\sqrt{m}}| 2\pi \sqrt{m}.
\]
By the splitting~\eqref{4}, 
\[
\begin{split}
&|\mathbb E[h(X_t,P_t) \int_{\tau_*}^{t/\sqrt{m}}\langle \bar V'' \cos(\tau \bar V''^{1/2}) \dot a(X_{t-\sqrt{m}\tau}),\partial_{X^\ell} a\rangle {\rm d} \tau]|\\
& \le 
|\mathbb E[h(X_t,P_t) \int_{\tau_*}^{t/\sqrt{m}} \sum_{\ell'} K_\infty^{\ell \ell'}(\tau) ( \dot{X}^{\ell'}_{t -\tau\sqrt{m}} - \dot{Y}^{\ell'}_{t -\tau\sqrt{m}}) {\rm d} \tau]|\\
&\qquad + |\mathbb E[h(X_t,P_t) \int_{\tau_*}^{t/\sqrt{m}} \sum_{\ell'} K_\infty^{\ell \ell'}(\tau) \dot{Y}^{\ell'}_{t -\tau\sqrt{m}} {\rm d} \tau]|\\
& \le 
 C \sum_{k=\lfloor \tau_*/\pi \rfloor}^{\lfloor t/(\sqrt{m}\pi) \rfloor} \left(
 \sqrt{m}  \int_{\tau_k}^{\tau_{k+1}} \frac{|\sin(\tau)|}{\tau}  \, {\rm d} \tau 
 + \left| \int_{\tau_k}^{\tau_{k+1}} \frac{\sin(\tau)}{\tau}  \, {\rm d} \tau \right| \right)\\
 &\le C(\sqrt{m}\log m^{-1} + \tau_*^{-1})
\end{split}
\]
where we used $\int_{\tau_k}^{\tau_{k+1}} \frac{\sin(\tau)}{\tau}  \, {\rm d} \tau=-
\int_{2k\pi}^{2(k+1)\pi} \frac{\cos(\tau)}{\tau^2}  \, {\rm d} \tau$,
and 
\[
\begin{split}
&\left| \mathbb E[h(X_t,P_t)\int_0^{\tau_*} 
\sum_{\ell'} K_{\infty}^{\ell \ell'}(\tau) (\dot X_{t - \sqrt{m}\tau }^{\ell'}-\dot X_{t }^{\ell'}) \, {\rm d}\tau]\right|\\
&\le 
\left| \mathbb E[h(X_t,P_t)\int_0^{\tau_*} 
\sum_{\ell'} K_{\infty}^{\ell \ell'}(\tau) (\dot X_{t - \tau m^{1/2}}^{\ell'} - \dot Y_{t - \tau m^{1/2}}^{\ell'}) \, {\rm d}\tau]\right| \\
&\qquad + \left| \mathbb E[h(X_t,P_t)\int_0^{\tau_*} 
\sum_{\ell'} K_{\infty}^{\ell \ell'}(\tau) (\dot Y_{t - \tau m^{1/2}}^{\ell'} - \dot X_{t}^{\ell'}) \, {\rm d}\tau]\right|\\
&\le C  \sqrt{m} \sum_{k=1}^{\lceil \tau_*/\pi \rceil} k^{-1} \\
&\le C\sqrt{m} \log \tau_*\, , 
\end{split}
\] 
and~\eqref{convRateH} follows by the same reasoning as in the proof of Lemma~\ref{fric_lem}.

\subsection{Dynamical systems}
For a given mass ratio $m$, the generalized Langevin equation of the heat bath dynamics takes the form
\[
\begin{split}
\dot{X}(t) &= P(t)\\
\dot{P}(t) &= -\nabla \lambda(X(t)) - \int_0^t K_{\infty}\left(\frac{t-s}{\sqrt{m}} \right) P(s) \, {\rm d}s + \zeta(t)\, ,
\end{split}
\]
where $\zeta(t)$ denotes a mean-zero Gaussian process with $\zeta^1 = \zeta^2 = \zeta^3$ and $\mathbb{E} [ \zeta^1(s) \zeta^1(t)] = T K_{\infty}^{11}( (t-s)/\sqrt{m})$, cf. \eqref{fluc_diss} and \eqref{cov_k}.
The associated Langevin dynamics is
\[
\begin{split}
\dot{X}_L(t) &= P_L(t)\\
\dot{P}_L(t) &= -\nabla \lambda(X_L(t)) - m^{1/2} \kappa P_L(t)  + (2 m^{1/2} \kappa T)^{1/2} \, \dot{W}(t)
\end{split}
\]
with $\kappa$ given by~\eqref{kappaNumEx}. We will compare the dynamical systems numerically 
for the initial data $X_L(0)= X(0) = \xi_X (1,1,1)$ and $P_L(0)= P(0)= \xi_P (1,1,1)$,
where $\xi_X$ and $\xi_P$ are independent identically distributed standard Gaussians that are sampled pathwise.
Due to the initial data and $\nabla \lambda(X) =X$, it holds that $X(t), P(t) \in \text{Span}( (1,1,1))$ for all $t\ge 0$. For this particular example, it therefore suffices to
study the reduced dynamics $(X^1,P^1)$ and $(X_L^1,P_L^1)$ rather than the 
respective 6 dimensional full systems. 
The respective reduced dynamics are equal in distribution to 
\[
\begin{split}
\dot{X}^1(t) &= P^1(t)\\
\dot{P}^1(t) &= -X^1(t) - 12 \pi \sqrt{m} \int_0^t \frac{\sin( (t-s)/\sqrt{m} )}{t-s} P^1(s) \, {\rm d}s + \zeta^1(t),
\end{split}
\]
and 
\begin{equation}\label{eq:exactLanDyn1D}
\begin{split}
\dot{X}_L^1(t) &= P_L^1(t)\\
\dot{P}_L^1(t) &= - X_L^1(t) - 6 \pi^2 \sqrt{m}  P_L^1(t)  + 2 \pi m^{1/4} T^{1/2} \, \dot{W}^1(t)\, .
\end{split}
\end{equation}

\subsection{Numerical integration schemes}
Langevin dynamics (St\"ormer--Verlet/Ornstein--Uhlenbeck~\cite{muller2015}):
\[
\begin{split}
P_{L,n+1/2}^* &= \exp(-  6 \pi^2 \sqrt{m} \Delta t/2) P_{L,n}^1 +    T^{1/2} 
\sqrt{\frac{1-\exp(-12 \pi^2 \sqrt{m} \Delta t/2)}{3  }}\xi_{2n-1} \\
P_{L,n+1/2}^1  &= P_{L,n+1/2}^* - X_{L,n}^1 \frac{\Delta t}{2} \\
X_{L,n+1}^1 &= X_{L,n}^1 + P_{L,n+1/2}^1 \Delta t\\
P_{L,n+1}^*  &= P_{L,n+1/2}^1 - X_{L,n+1}^1 \frac{\Delta t}{2} \\
P_{L,n+1}^1 &= \exp(-  6 \pi^2 \sqrt{m} \Delta t/2) P_{L,n+1}^* +    T^{1/2} 
\sqrt{\frac{1-\exp(-12 \pi^2 \sqrt{m} \Delta t/2)}{3  }}\xi_{2n} \, ,
\end{split}
\]
where $\xi_n$ is a sequence of independent and identically distributed standard normals.
The scheme is motivated from the splitting method with symplectic integration 
of the Hamiltonian system
\[
\dot{X}^1_L(t) = P^1_L(t), \qquad  \dot{P}_L^1(t) = - X_L^1(t)
\]
and exact solution of the Ornstein--Uhlenbeck equation
\[
\dot{P}_L^1(t) = - 6 \pi^2 \sqrt{m}  P_L^1(t) +  2 \pi m^{1/4} T^{1/2} \, \dot{W}^1(t),
\]
cf.~\cite{muller2015}. 

For the heat bath dynamics we construct a splitting scheme which for a uniform 
mesh $t_{k} = k \Delta t$ computes the position at every 
timestep ($X^1(t_0), X^1(t_1),\ldots$) and the momentum at every 
half-timestep ($P^1(t_0),P^1(t_{1/2}), P^1(t_1), \ldots$). 
The damping term's integral is approximated as follows: 
\begin{equation}\label{eq:dampingApprox}
\begin{split}
&\int_0^{t_{n+1}} \frac{\sin((t_{n+1}-s)/\sqrt{m})}{t_{n+1}-s} P^1(s) {\rm d}s\\
&\approx \sum_{k=0}^n {P^1}(t_{k}) \int_{t_k}^{t_{k+1/2}} \frac{\sin((t_{n+1}-s)/\sqrt{m})}{t_{n+1}-s} {\rm d} s\\
&\quad + \sum_{k=0}^n {P^1}(t_{k+1/2}) \int_{t_{k+1/2}}^{t_{k+1}} \frac{\sin((t_{n+1}-s)/\sqrt{m})}{t_{n+1}-s} {\rm d} s\\
&= \sum_{k=0}^n \Big[  P^1(t_{k}) (\Si(t_{n+1-k}/\sqrt{m}) - \Si(t_{n+1/2-k}/\sqrt{m})) \\
& \qquad \qquad + P^1(t_{k+1/2}) (\Si(t_{n+1/2-k}/\sqrt{m}) - \Si(t_{n-k}/\sqrt{m})) \Big]\, ,
\end{split}
\end{equation}
where the last equality follows from 
\[
\int_{a}^{b} \frac{\sin((t_{n+1}-s)/\sqrt{m})}{t_{n+1}-s} {\rm d} s
= \int_{(t_{n+1}-b)/\sqrt{m}}^{(t_{n+1}- a)/\sqrt{m}} \frac{\sin(s)}{s} \, {\rm d}s 
\]
and
\[
\Si(t) := \int_{0}^t \frac{\sin(s)}{s} \, {\rm ds}\, . 
\]
Introducing the notation 
\[
\Delta \Si_{\ell} := \Si(t_{\ell+1/2}/\sqrt{m}) - \Si(t_{\ell}/\sqrt{m})\, ,
\]
the approximation~\eqref{eq:dampingApprox} takes the compact form
\[
\int_0^{t_{n+1}} \frac{\sin((t_{n+1}-s)/\sqrt{m})}{t_{n+1}-s} P^1(s) {\rm d}s
\approx 
\sum_{k=0}^n P^1(t_{k}) \Delta \Si_{n+1/2-k} + P^1(t_{k+1/2}) \Delta\Si_{n-k} \, ,
\]
and, similarly, 
\[
\int_0^{t_{n+1/2}} \frac{\sin((t_{n+1/2}-s)/\sqrt{m})}{t_{n+1/2}-s} P^1(s) {\rm d}s
\approx 
\sum_{k=0}^n P^1(t_{k}) \Delta \Si_{n-k} +  \sum_{k=0}^{n-1} P^1(t_{k+1/2}) \Delta\Si_{n+1/2-k}\, .
\]
We make use of the above approximations of the damping term integral
in the following splitting scheme for the heat bath dynamics: 
\begin{equation}\label{eq:heatBathDynScheme}
\begin{split}
P^1_{n+1/2}  &= P^1_n - \frac{\Delta t}{2}\prt{ X^1_n - \zeta^1(t_n)
+ 12 \pi \sqrt{m}  \prt{\sum_{k=0}^n P^1_k \Delta \Si_{n-k} +  \sum_{k=0}^{n-1} P^1_{k+1/2} \Delta\Si_{n+1/2-k}}}\\
X^1_{n+1} &= X^1_{n} +P^1_{n+1/2} \Delta t\\
P^1_{n+1}  &= P^1_{n+1/2} - \frac{\Delta t}{2}\prt{X^1_{n+1} - \zeta^1(t_{n+1}) 
+ 12 \pi \sqrt{m} \sum_{k=0}^n P^1_k \Delta \Si_{n+1/2-k} + P^1_{k+1/2} \Delta\Si_{n-k}}\, .
\end{split}
\end{equation}
The mean-zero Gaussian vector $ \boldsymbol{\zeta} = (\zeta^1(t_0), \zeta^1(t_1), \ldots, \zeta^1(t_N))$ is sampled 
by computing the square root of the Toeplitz matrix with first row vector 
\[
(T K_\infty(0^+/\sqrt{m}), T K_\infty(t_1/\sqrt{m}), \ldots, T K_\infty(t_N/\sqrt{m}))
\]
and multiplying the square root matrix, say $\sqrt{K}$, with an $(N+1)-$vector $\xi$ of iid standard normals
components: $\boldsymbol{\zeta} = \sqrt{K} \xi$.
See~\cite{kroese2013} for further details on sampling of Gaussian 
processes,~\cite{brunger1984stochastic,skeel2002impulse,leimkuhler2016molecular,mattingly2002ergodicity,bou2010long,abdulle2015long,lelievre2016partial} for 
numerical methods for Langevin dynamics and~\cite{baczewski2013numerical, hall2016uncertainty,kupferman2004fractional} and~\cite[Chapter 8.7]{leimkuhler2016molecular} for an alternative numerical method for
generalized Langevin equations based on a truncated prony series approximation of the kernel $K_\infty$.

\subsection{Observations}
By setting the temperature to $T=3$, the stationary distribution of the exact Langevin dynamics~\eqref{eq:exactLanDyn1D} becomes $N(0,I)$ for any $m>0$, with $I$ denoting the identity matrix in $\mathbb{R}^2$. In order to reduce the computational 
challenges of long time numerical integration, we sample the initial data from this stationary, i.e., $(\xi_X, \xi_P) \sim N(0,I)$, as we assume this yields initial data for the numerical dynamics $(X^1_{L,0}, P^1_{L,0})$ and $(X^1_{0}, P^1_{0})$ that are very close to their 
respective stationary distributions. The numerical computations are performed 
using a timestep $\Delta t=0.005$, $N=4000$ integration steps
relating to the final time $t_N=20$, and $M=2\times10^{5}$ solution realizations
of the respective dynamics. 
Figures~\ref{fig:histX} and~\ref{fig:histP} show good correspondence between the final time marginal distributions of the heat bath dynamics $(X^1_N, P^1_N)$ and the Langevin dynamics $(X^1_{L,N},P^1_{L,N})$ over a range of $m$-values (all marginals being approximately $N(0,1)$-distributed).
Figures~\ref{fig:acX} and~\ref{fig:acP} show that the auto-correlation for both the position and the momentum 
of the heat bath dynamics converges to the corresponding ones for the Langevin dynamics as $m\to 0^+$.
The computations of the auto-correlation functions are made under the assumption that both kinds 
of dynamics are wide-sense stationary. For any $m>0$ this property holds for the Langevin dynamics, 
and in the limit $m\to 0^+$ it also holds for the heat bath dynamics. 
\begin{figure}[H]
\centering
\includegraphics[width=0.8\textwidth]{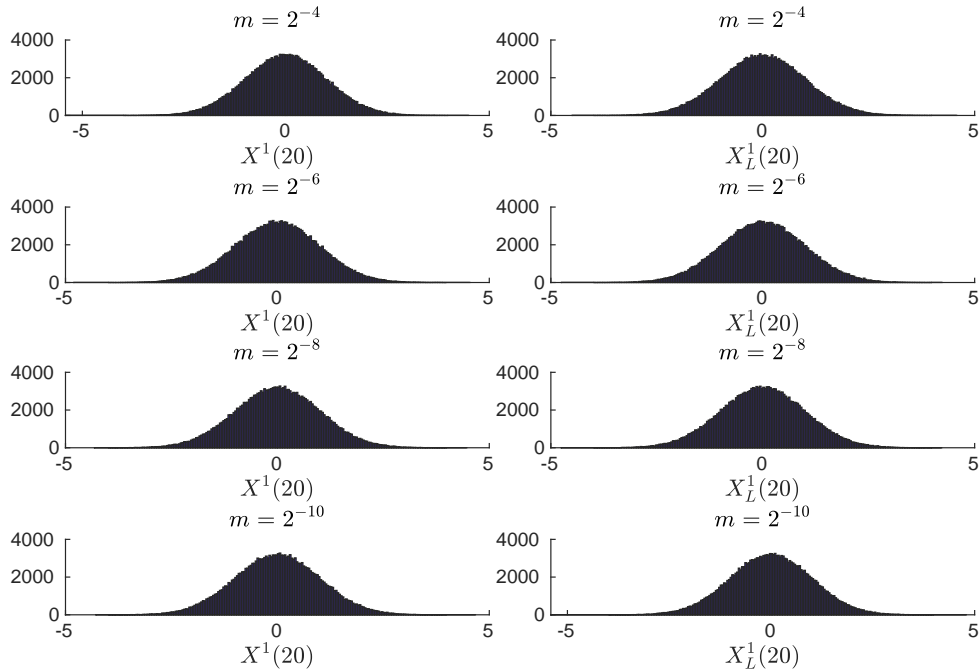}
\caption{(Left column) histogram for the final time position of the heat bath dynamics for a series of $m$-values
and (right column) corresponding histograms for the Langevin dynamics.
 }\label{fig:histX}
\end{figure}

\begin{figure}[H]
\centering
\includegraphics[width=0.8\textwidth]{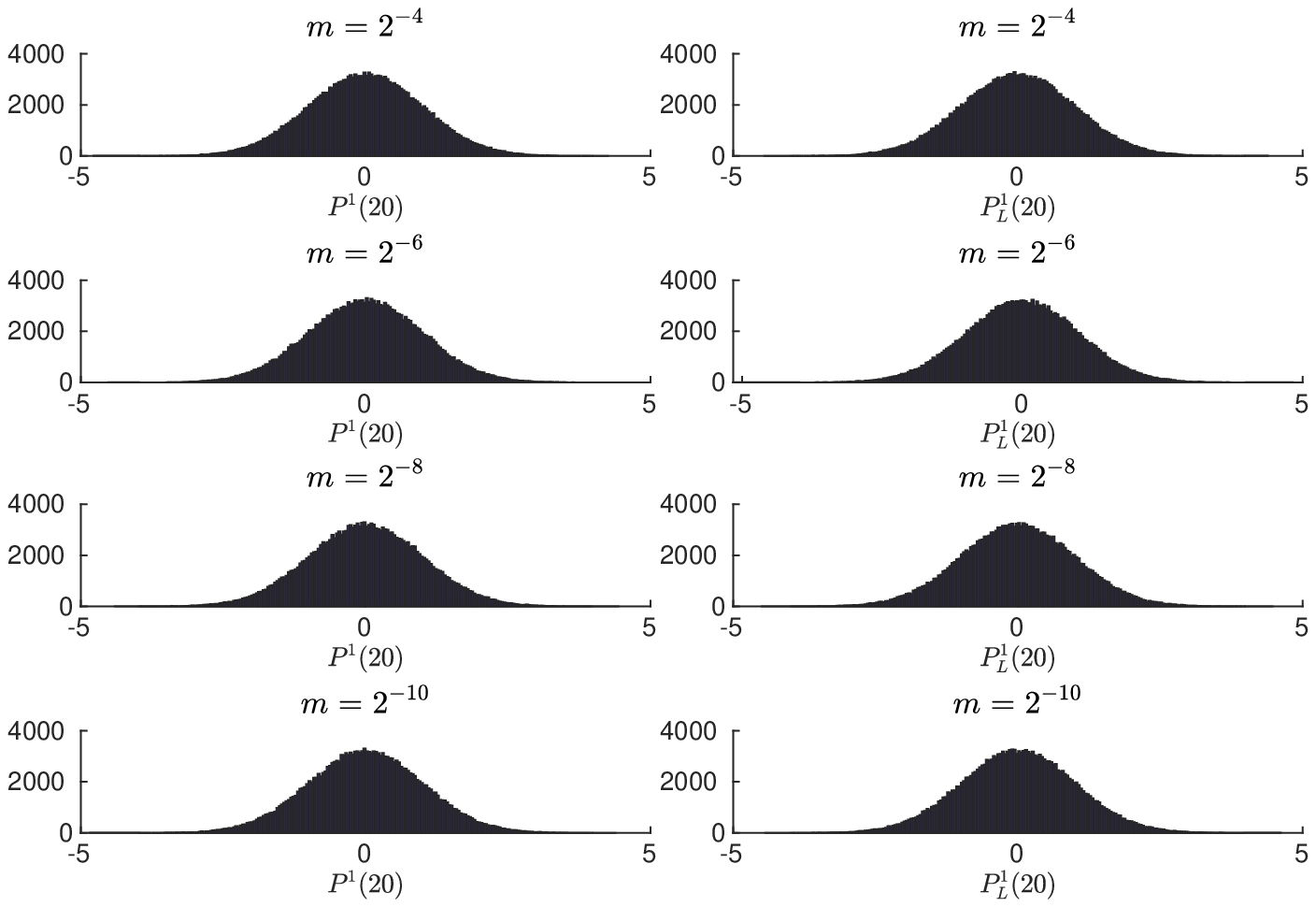}
\caption{
(Left column) histogram for the final time momentum of the heat bath dynamics 
for a series of $m$-values and (right column) corresponding histograms for the Langevin dynamics.
 }\label{fig:histP}
\end{figure}

\begin{figure}[H]
\centering
\includegraphics[width=0.82\textwidth]{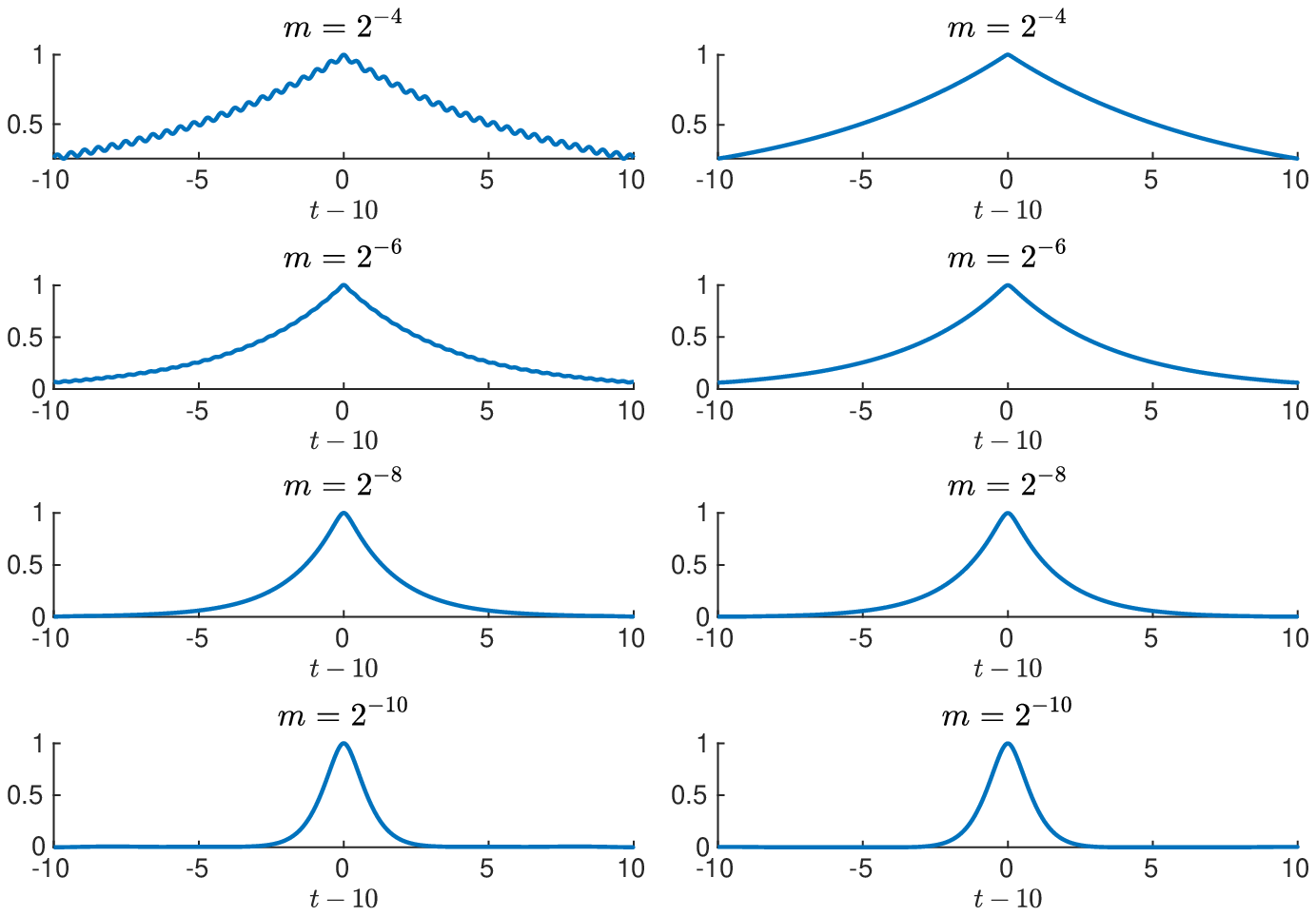}
\caption{
(Left column) heat bath dynamics auto-correlation function $\mathbb E[X^1(t)X^1(10)]$
for a series of $m$-values 
and (right column) corresponding Langevin dynamics auto-correlation functions $\mathbb E[X^1_L(t)X^1_L(10)]$. }\label{fig:acX}
\end{figure}

\begin{figure}[H]
\centering
\includegraphics[width=0.82\textwidth]{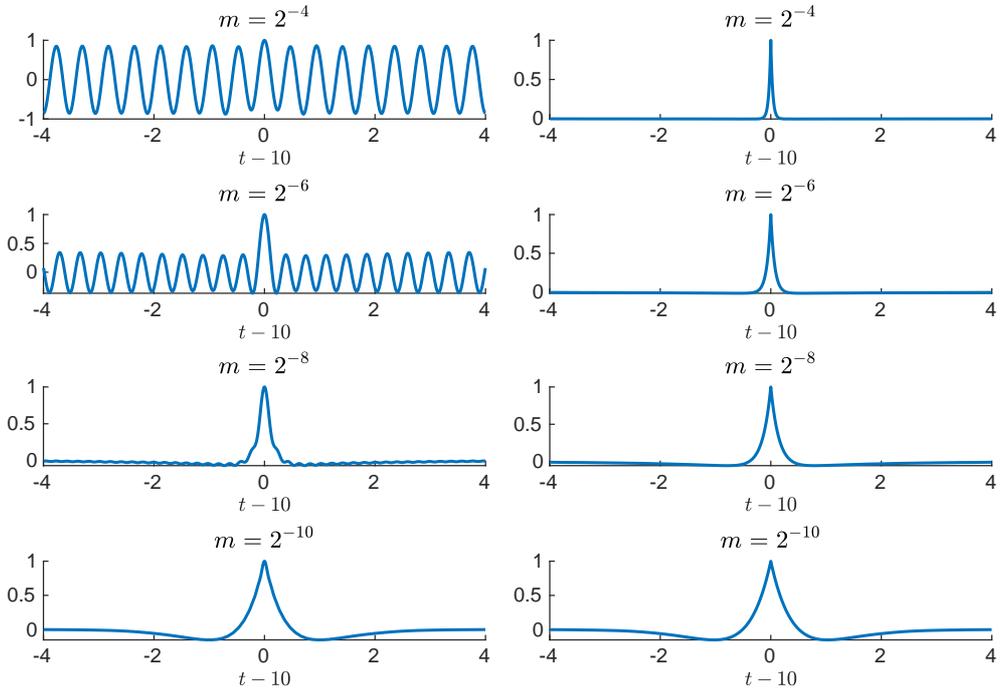}
\caption{
(Left column) heat bath dynamics auto-correlation function $\mathbb E[P^1(t)P^1(10)]$
for a series of $m$-values 
and (right column) corresponding Langevin dynamics auto-correlation functions 
$\mathbb E[P^1_L(t)P^1_L(10)]$. }\label{fig:acP}
\end{figure}

In addition to our above observations, we believe it would be of great interest to obtain numerical 
verification for the heat bath dynamics weak convergence rate $\mathcal{O}(m^{1/2})$ in~\eqref{larger_err}. But, most likely due to computational constraints, we are unable to achieve this currently since even at
the quite computationally demanding level of generating $M=200000$ heat bath dynamics sample paths,
it seems that the sample error dominates errors pertaining to the parameter $m$.

\end{document}